\RequirePackage{amsmath}
\documentclass[runningheads]{llncs}

\usepackage{graphicx}
\usepackage{tikz}
\usepackage{amssymb}
\usepackage{ebproof}
\usepackage{proof}
\usepackage{newunicodechar}
\usepackage{xspace}
\usepackage{multirow}
\usepackage[hidelinks]{hyperref}
\usepackage{enumerate}
\usepackage{array}

\usepackage{color}

\usepackage{etoolbox}
\AtEndEnvironment{proof}{\phantom{}\hfill$\blacksquare$} 

\newcommand{\goodbox}{\hspace{-.6ex}\text{
    \tikz[baseline=-.6ex, rounded corners=.01ex, line width=.1ex]
    {\draw (-.6ex,-.6ex) rectangle (.6ex,.6ex);}}\kern.2ex}
\renewcommand{\Box}{\goodbox}

\newcommand{\xBox}{\Box^{-1}} 

\newcommand{\gooddiamond}{\hspace{-.6ex}\text{
    \tikz[baseline=-.6ex, rounded corners=.01ex, rotate=45, line width=.1ex]
    {\draw (-.5ex,-.5ex) rectangle (.5ex,.5ex);}}\kern.2ex}
\newcommand{\Diam}{\gooddiamond}

\newcommand{\xDiam}{\Diam^{-1}}

\renewcommand{\phi}{\varphi}

\newcommand{\lang}{\mathcal{L}}
\newcommand{\sys}[1]{\mathsf{#1}}
\newcommand{\sysrule}[1]{(\text{#1})}
\newcommand{\Prop}{{\sf{Prop}}}
\newcommand{\openboxes}{\xBox}
\newcommand{\Ra}{\Rightarrow}
\newcommand{\To}{\Rightarrow}

\usepackage{fontawesome5}

\newcommand{\K}{\sys{K}}
\newcommand{\iK}{\sys{iK}}
\newcommand{\IK}{\sys{IK}}
\newcommand{\IL}{\sys{IL}}
\newcommand{\CK}{\sys{CK}}
\newcommand{\WK}{\sys{WK}}

\newcommand{\CKH}{\sys{CKH}}
\newcommand{\WKH}{\sys{WKH}}
\newcommand{\Kbax}{\hyperlink{ax:Kb}{\text{K}_{\Box}}}
\newcommand{\Kdax}{\hyperlink{ax:Kd}{\text{K}_{\Diam}}}
\newcommand{\Nax}{\hyperlink{ax:N}{\text{N}}}
\newcommand{\Axr}{\hyperlink{rule:Ax}{\sysrule{Ax}}}
\newcommand{\El}{\hyperlink{rule:El}{\sysrule{El}}}
\newcommand{\MP}{\hyperlink{rule:MP}{\sysrule{MP}}}
\newcommand{\Nec}{\hyperlink{rule:Nec}{\sysrule{Nec}}}

\newcommand{\Gfourip}{\sys{G4ip}}
\newcommand{\GfourCK}{\sys{G4CK}}
\newcommand{\GfourWK}{\sys{G4WK}}
\newcommand{\IdP}{\hyperlink{rule:IdP}{\sysrule{IdP}}}
\newcommand{\Id}{\hyperlink{rule:Id}{\sysrule{Id}}}
\newcommand{\botL}{\hyperlink{rule:botL}{\sysrule{$\bot$L}}}
\newcommand{\impR}{\hyperlink{rule:impR}{\sysrule{$\rightarrow$R}}}
\newcommand{\impL}{\hyperlink{rule:impL}{\sysrule{$\rightarrow$L}}}
\newcommand{\imppL}{\hyperlink{rule:imppL}{\sysrule{$p\rightarrow$L}}}
\newcommand{\impandL}{\hyperlink{rule:impandL}{\sysrule{$\land\!\rightarrow$L}}}
\newcommand{\imporL}{\hyperlink{rule:imporL}{\sysrule{$\lor\!\rightarrow$L}}}
\newcommand{\impimpL}{\hyperlink{rule:impimpL}{\sysrule{$\rightarrow\rightarrow$L}}}
\newcommand{\impboxL}{\hyperlink{rule:impboxL}{\sysrule{$\Box\!\rightarrow$L}}}
\newcommand{\impdiamL}{\hyperlink{rule:impdiamL}{\sysrule{$\Diam\!\rightarrow$L}}}
\newcommand{\andL}{\hyperlink{rule:andL}{\sysrule{$\land$L}}}
\newcommand{\andR}{\hyperlink{rule:andR}{\sysrule{$\land$R}}}
\newcommand{\orRone}{\hyperlink{rule:orR}{\sysrule{$\lor$R$_1$}}}
\newcommand{\orRtwo}{\hyperlink{rule:orR}{\sysrule{$\lor$R$_2$}}}
\newcommand{\orL}{\hyperlink{rule:orL}{\sysrule{$\lor$L}}}
\newcommand{\BoxR}{\hyperlink{rule:BoxR}{\sysrule{$\Box$R}}}
\newcommand{\DiamL}{\hyperlink{rule:DiamL}{\sysrule{$\Diam$L}}}
\newcommand{\DiamLW}{\hyperlink{rule:DiamLW}{\sysrule{$\Diam$L'}}}
\newcommand{\wk}{\hyperlink{rule:wk}{\sysrule{wk}}}
\newcommand{\cntr}{\hyperlink{rule:cntr}{\sysrule{cntr}}}
\newcommand{\cut}{\hyperlink{rule:cut}{\sysrule{cut}}}

\newcommand{\symApCK}{\mathsf{A}_p^{\CK}}
\newcommand{\symEpCK}{\mathsf{E}_p^{\CK}}
\newcommand{\symApWK}{\mathsf{A}_p^{\WK}}
\newcommand{\symEpWK}{\mathsf{E}_p^{\WK}}
\newcommand{\Ap}[1]{\mathsf{A}_{p}(#1)}
\newcommand{\Ep}[1]{\mathsf{E}_{p}(#1)}
\newcommand{\callAp}[1]{\mathsf{\mathcal{A}_p}(#1)}
\newcommand{\callEp}[1]{\mathsf{\mathcal{E}_p}(#1)}
\newcommand{\pvf}[1]{\text{Vars}\,(#1)}
\newcommand{\pvs}[1]{\text{Vars}\,(#1)}

\def\BaseUrl{https://ianshil.github.io/CK}
\newcommand{\rocqdoc}[1]{\href{\BaseUrl/#1}{\raisebox{-0.6mm}{\includegraphics[height=0.8em]{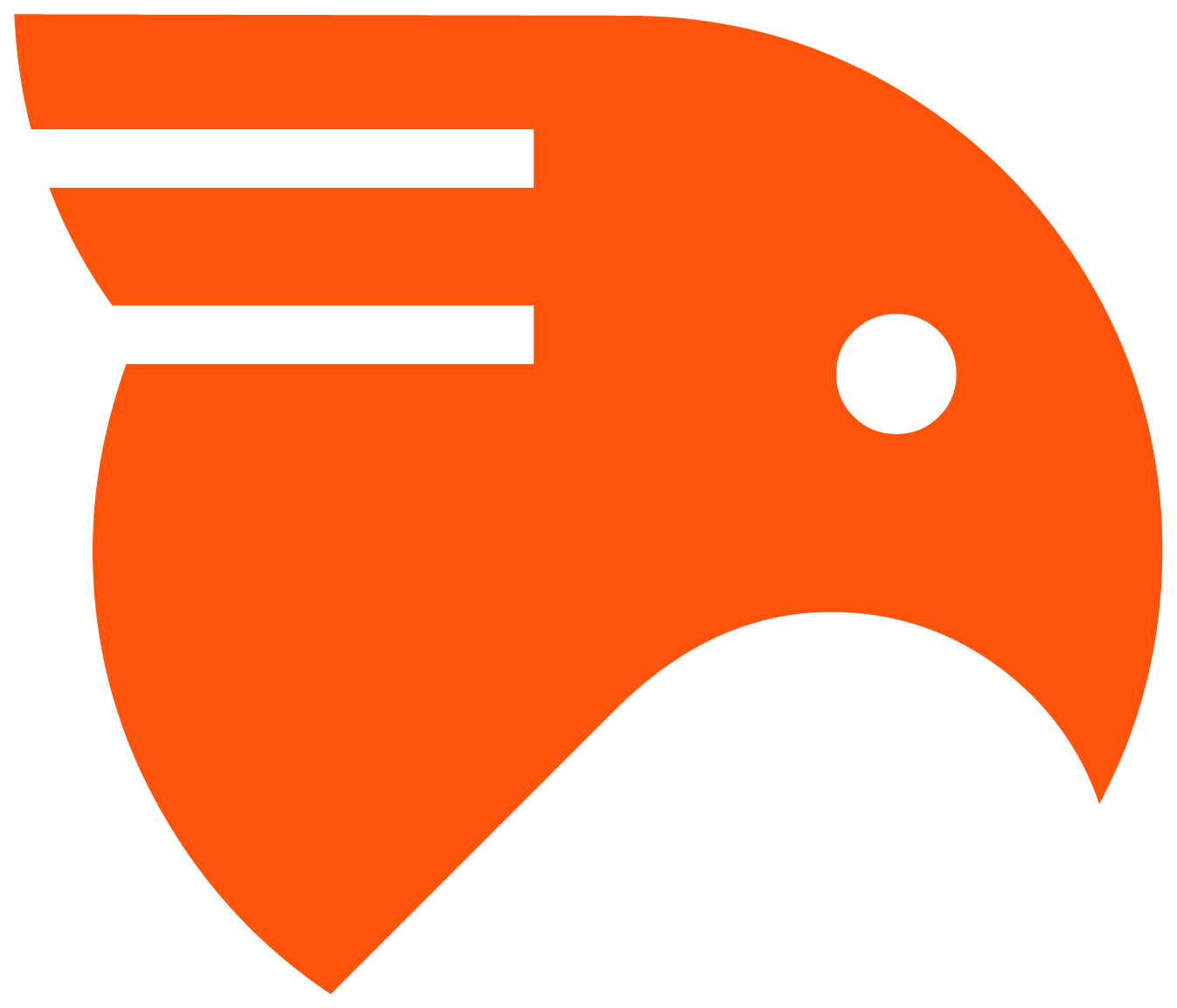}}}}

\title{Uniform interpolation with constructive diamond}
\author{Iris van der Giessen\inst{1}\orcidID{0009-0008-4908-2496}
\and 
Ian Shillito\inst{2,3}\orcidID{0009-0009-1529-2679}}
\authorrunning{I. van der Giessen and I. Shillito}

\institute{University of Amsterdam, Amsterdam, The Netherlands\\
\email{i.vandergiessen@uva.nl} \and
University of Birmingham, Birmingham, United Kingdom 
\and
University of Melbourne, Melbourne, Australia \\
\email{ian.shillito@unimelb.edu.au}}

\begin{document}

\maketitle

\begin{abstract}
Uniform interpolation is a strong form of interpolation providing an interpretation of propositional quantifiers within a propositional logic. Pitts’ seminal work establishes this property for intuitionistic propositional logic relying on a sequent calculus in which naïve backward proof-search terminates. This constructive approach has been adapted to a wide range of logics, including intuitionistic modal logics. Surprisingly, no intuitionistic modal logic with independent box and diamond has yet been shown to satisfy uniform interpolation. We fill in this gap by proving the uniform interpolation property for Constructive K (CK) and Wijesekera's K (WK).  We build on Pitts' technique by exploiting existing terminating calculi for CK and WK, which we prove to eliminate cut, and formalise all our results in the proof assistant Rocq. Together, our results constitute the first positive uniform interpolation results for intuitionistic modal logics with diamond.
\keywords{Uniform Interpolation \and Proof Theory \and Intuitionistic Modal Logic \and Rocq.}
\end{abstract}

\setcounter{footnote}{0}

\section{Introduction}

Intuitionistic modal logics form a rich class of logics that formalise modal reasoning over an intuitionistic propositional base. They pop up in different forms and applications and are sensitive to many design choices. 
When considering intuitionistic modal logics with both $\Box$ and $\Diam$, a common feature is that $\Box$ and $\Diam$ are not interdefinable (this in contrast to classical modal logics). One traditionally identifies two main streams to define such logics: \textit{Intuitionistic modal logic} and \textit{Constructive modal logic}. Intuitionistic modal logics were early investigated by, e.g., Fischer Servi \cite{FischerServi77,FischerServi84}, Plotkin and Stirling \cite{Plotkin_Stirling86}, and Ewald \cite{Ewald1986} and are motivated by their algebraic connection to classical bi-modal logics \cite{FischerServi77,WolterZakharyaschev97,WolterZakharyaschev99a,WolterZakharyaschev99b} and their interpretation via the standard translation into first-order intuitionistic logic \cite{Simpson94PhD}. The base system in this setting is Intuitionistic $\K$, here denoted~$\IK$, as a counterpart to classical modal logic $\K$. Constructive modal logics are mostly motivated by computational applications. Constructive $\K$, here denoted $\CK$, and its modal extensions are motivated by their Curry-Howard interpretation in type theory and their categorical semantics \cite{BiePai00,AleMenPaiRit01,BeldePRit01,Kav16}. A closely related constructive modal logic is Wijesekera's system introduced to model reasoning in dynamic systems \cite{Wijesekera90}. We denote its propositional part by $\WK$.\footnote{Wijesekera's full system could be viewed as a more general intuitionistic modal logic than the Constructive Concurrent Dynamic Logic $\sf{CCDL}$ from \cite{WijesekeraNerode05}. Some authors also denote $\WK$ as $\sf{CCDL}$ (e.g., in \cite{DalGreOli21}), although, originally, $\sf{CCDL}$ was introduced over a richer syntax including program constructors.} Recently, new approaches have been proposed to define intuitionistic versions of classical~$\K$ \cite{Balbiani_et_al24a,Balbiani_et_all24b}. In this paper we concentrate on the well-studied constructive modal logics $\CK$ and $\WK$.

The research on these two logics is extensive and revealed some surprises in recent years. Axiomatically, the two logics form extensions of intuitionistic propositional logic $\IL$ with the following modal rule and (subset) of modal axioms below. These axioms follow classically from the usual normality axiom $\Kbax$, but are independent in an intuitionistic setting. We define $\CK = \IL + \{\Kbax,\Kdax\}$ and $\WK = \CK + \{\Nax\}$, and note that $\CK \subset \WK$.
\begin{center}
$
\infer[\scriptstyle\Nec]{\Gamma\vdash\Box\varphi}{\emptyset\vdash\varphi}
\qquad \qquad
\begin{array}{ll}
    \Kbax & {\Box(\varphi\rightarrow\psi)\rightarrow(\Box\varphi\rightarrow\Box\psi)}\\
   \Kdax & {\Box(\varphi\rightarrow\psi)\rightarrow(\Diam\varphi\rightarrow\Diam\psi)}\\
    \Nax & {\neg\Diam\bot}
\end{array}
$
\end{center}
For semantic characterisations of the logics one can consult~\cite{GroShiClo25}.
Recent studies~\cite{DasMar22-blog,DasMar23,GroShiClo25,DasGroShi25-blog} surprisingly resolved a common misbelief that the $\Diam$-free fragments of Constructive and Intuitionistic modal logics coincide: those for $\CK$ and~$\WK$ coincide with the $\Box$-only logic $\iK$, but the one for $\IK$ is strictly stronger. 

We are interested in the interpolation properties for $\CK$ and $\WK$. Interpolation properties form a central theme in logic due to their good mathematical properties and widespread applications in computer science (see \cite{BookCraigInterpolation} for a recent overview on foundations, methods and applications of interpolation). \textit{Craig interpolation} states that if $\phi$ entails $\psi$, there is a formula $\theta$ in the common vocabulary of $\phi$ and~$\psi$ such that $\phi$ entails $\theta$ and $\theta$ entails $\psi$. The formula $\theta$ is called the \textit{interpolant} and explains the reason why $\phi$ entails $\psi$. \textit{Uniform interpolation} is stronger, 
as the interpolant only depends on $\phi$ and uniformly works for any $\psi$ in a specified vocabulary. This property makes it possible to provide an interpretation of propositional quantifiers within the propositional base \cite{Pit92}.

Proof-theoretic methods are ideal for establishing metalogical results such as interpolation.
Indeed, the Craig interpolation property has been established for $\WK$ in Wijesekera's original paper using a sequent calculus \cite{Wijesekera90}. In~\cite{DalGreOli21}, terminating sequent calculi for $\CK$ and $\WK$ were designed, but the uniform interpolation question was left open. Even though other systems are designed for $\CK$ and~$\WK$ including a sequent system, a natural deduction system \cite{BeldePRit01}, a nested sequent calculus \cite{ArisakaDasStrassburger15} and a focused 2-sequent calculus \cite{Mendler11} for $\CK$, and a tableaux calculus for the richer version $\sf{CCDL}$ of $\WK$   \cite{WijesekeraNerode05}, neither Craig interpolation for $\CK$ nor uniform interpolation for $\CK$ or $\WK$ has been investigated. Most surprising, to the best of our knowledge, no other interpolation result is known for \textit{any} Constructive or Intuitionistic modal logic with independent $\Box$ and $\Diam$~\cite{Kur26}, also not via semantic means, except for one. 
We came to know via personal communication with N.~Bezhanishvili that the intuitionistic modal logic $\sf{MIPC}$ \cite{Prior57TimeModality,Bull65,Bull66}, which can be viewed as an $\sf{S5}$-extension of $\IK$, does not have the Craig interpolation property.

In this paper we focus on constructive methods towards uniform interpolation. The seminal work of Pitts provides a proof-theoretic proof of uniform interpolation for $\IL$ using a strongly terminating sequent calculus. It is now a fruitful technique for proving uniform interpolation among many types of modal logic: classical modal logics~\cite{Bil06,Bil22,FerGieGooShi24}, intuitionistic modal logics~\cite{Iemhoff19b,FerGieGooShi24}, modal substructural logics and linear logics~\cite{AliDerOno2014,AkbarTabatabai_Jalali25preprint}, and non-normal modal and conditional logics \cite{AkbarTabatabai_Iemhoff_Jalali24}.  However, these studies only treat mono-modal logics (with $\Box$), but do not cover intuitionistic modal logics with independent $\Box$ and~$\Diam$. 

\paragraph{Our contributions} We positively answer the uniform interpolation question for $\CK$ and~$\WK$. This question is specifically posed by Dalmonte, Grellois and Olivetti in~\cite{DalGreOli21} in which they designed strongly terminating sequent calculi for~$\CK$ and~$\WK$. Our work can be seen as a follow-up in which we 
(1) provide single-succedent variants of these calculi and provide termination and cut-elimination proofs, 
(2) provide a \textit{constructive} proof of uniform interpolation à la Pitts for~$\CK$ and~$\WK$, and 
(3) contribute to The Rocq Prover~\cite{rocq} library for logics between~$\CK$ and~$\IK$~\cite{GroShiClo25} by formalising all our results. All results described in this paper are accompanied by a clickable symbol ``\rocqdoc{}" leading to their formalisation.

Our work establishes the first positive uniform interpolation result among intuitionistic modal logics with independent $\Box$ and $\Diam$. Since uniform interpolation is stronger than Craig interpolation, we immediately obtain the Craig interpolation property for $\CK$ and $\WK$ as well. The constructive approach in Rocq provides us with a uniform interpolation calculator for $\WK$ and $\CK$ already available for some classical and intuitionistic modal logics \cite{Fer23,FerGieGooShi24}.

\section{Preliminaries}\label{sec:prelims}

In this section we fix the syntax and axiomatic calculi
for $\CK$ and $\WK$ as presented in de Groot, Shillito and Clouston's work~\cite{GroShiClo25}. 
For the mechanisation of most of the elements of this section,
we refer to their paper.

  Using a countably infinite set of propositional variables~$\Prop=\{p,q,r, \dots\}$, we define
  the language $\lang$ via the following grammar in BNF notation (\rocqdoc{Syntax.im_syntax.html\#form}):
  \begin{equation*}
    \varphi ::= p\in\Prop \mid \bot \mid \varphi \land \varphi \mid \varphi
    \lor \varphi \mid \varphi \rightarrow \varphi \mid \Box\varphi \mid \Diam\varphi
  \end{equation*}
  We abbreviate $\neg\phi:=\phi\rightarrow\bot$ and $\top:=\neg\bot$.
  We use Greek lowercase letters, e.g. $\phi, \psi, \chi$ and $\delta$, to
  denote formulas, and Greek uppercase letters, e.g. $\Gamma, \Delta, \Phi, \Psi$,
  for multisets of formulas.
  For such a multiset $\Gamma$ we define the two multisets
  $\Box\Gamma := \{\Box\varphi \mid \varphi\in\Gamma\}$
  and $\Box^{-1}\Gamma := \{ \phi \mid \Box\phi \in \Gamma \}$ and similarly for
  $\Diam\Gamma$ and $\Diam^{-1}\Gamma$.
  We also write $\Gamma\uplus\Delta$ for the \emph{disjoint union} of multisets $\Gamma$ and $\Delta$.
  If $\Gamma$ is finite, $\bigvee\Gamma$ denotes the disjunction of
  all formulas in $\Gamma$ 
  and $\bigwedge \Gamma$ denotes the conjunction of all formulas in $\Gamma$, with the convention that $\bigvee \emptyset = \bot$ and $\bigwedge \emptyset = \top$.

To finish with the syntax, we define a notion which 
we use to show termination of backward proof-search in
our sequent calculi.

\begin{definition}[\rocqdoc{Sequent.syntax_facts.html\#weight}]\label{def:weight}
The \textit{weight} $w(\varphi)$ of a formula is defined as follows.
\begin{center}
\begin{tabular}[c]{r c l}
$w(\bot)=w(p)$ & = & $1$
\\
$w(\psi\lor\chi)=w(\psi\rightarrow\chi)$ & = & $w(\psi) + w(\chi) + 1$
\\
$w(\psi\land\chi)$ & = & $w(\psi) + w(\chi) + 2$
\\
$w(\Box \psi)=w(\Diam\psi)$ & = & $w(\psi) + 1$
\\
\end{tabular}
\end{center}
\end{definition}

We provide generalised Hilbert calculi for $\CK$ and $\WK$,
i.e.~axiomatic calculi manipulating \emph{consecutions} of the shape
$\Gamma\vdash\phi$ where $\Gamma$ is a \emph{set} of formulas.
The calculus $\CKH$ for~$\CK$ extends the one for intuitionistic logic $\IL$,
with its axioms and rules, with the necessitation rule \Nec~and axioms $\Kbax$ and $\Kdax$.
The calculus $\WKH$ for~$\WK$ is nothing but the calculus
$\CKH$ augmented with the axiom $\Nax$.
All the axioms and rules just mentioned are presented in Figure~\ref{fig:Hilbert}.
We write $\Gamma\vdash_{\sys{H}}\varphi$ if $\Gamma\vdash\varphi$ is provable in $\sys{H}$,
for $\sys{H}\in\{\CKH,\WKH\}$.

\begin{figure}[t]
\begin{center}
\begin{tabular}{c}
\textbf{Axioms}\\
\end{tabular}
\end{center}
\vspace{-0.8cm}

\begin{center}
\begin{tabular}{l l @{\hspace{0.3cm}} l @{\hspace{0.1cm}} l}
$A_{1}$ & $\varphi\to(\psi\to\varphi)$ & $A_{7}$ & $(\varphi\land\psi)\to\psi$ \\
$A_{2}$ & $(\varphi\to(\psi\to\chi))\to(\varphi\to\psi)\to(\varphi\to\chi)$ & $A_{8}$ &  $(\varphi\rightarrow\psi)\rightarrow(\varphi\rightarrow\chi)\rightarrow(\varphi\rightarrow(\psi\land\chi))$ \\
$A_{3}$ & $\varphi\rightarrow(\varphi\lor\psi)$ & $A_{9}$ & $\bot\rightarrow\varphi$ \\
$A_{4}$ & $\psi\rightarrow(\varphi\lor\psi)$ & $\Kbax$ & \hypertarget{ax:Kb}{$\Box(\varphi\rightarrow\psi)\rightarrow(\Box\varphi\rightarrow\Box\psi)$} \\
$A_{5}$ & $(\varphi\rightarrow\chi)\rightarrow(\psi\rightarrow\chi)\rightarrow((\varphi\lor\psi)\rightarrow\chi)$ & $\Kdax$ & \hypertarget{ax:Kd}{$\Box(\varphi\rightarrow\psi)\rightarrow(\Diam\varphi\rightarrow\Diam\psi)$} \\
$A_{6}$ & $(\varphi\land\psi)\to\varphi$ & $\Nax$ & \hypertarget{ax:N}{$\neg\Diam\bot$} \\
\end{tabular}
\end{center}
\vspace{-0.7cm}

\begin{center}
\begin{tabular}{c}
\textbf{Rules of Inference}\\
\end{tabular}
\end{center}
\vspace{-0.5cm}

\begin{center}
\begin{tabular}{c@{\hspace{1.5cm}}c}
\hypertarget{rule:Ax}{$
\inferLineSkip=3pt
\infer[\scriptstyle\Axr]{\Gamma\vdash\varphi}{\varphi\text{ is an instance of an axiom }}
$} & 
\hypertarget{rule:El}{$
\inferLineSkip=3pt
\infer[\scriptstyle\El]{\Gamma\vdash\varphi}{\varphi\in\Gamma}
$} \\
 & \\
\hypertarget{rule:Nec}{$
\inferLineSkip=3pt
\infer[\scriptstyle\Nec]{\Gamma\vdash\Box\varphi}{\emptyset\vdash\varphi}
$} &
\hypertarget{ax:MP}{$
\inferLineSkip=3pt
\infer[\scriptstyle\MP]{\Gamma\vdash\psi}{
	\Gamma\vdash\varphi
	&
	\Gamma\vdash\varphi\rightarrow\psi}
$}
\end{tabular}
\end{center}
\caption[]{Generalised Hilbert calculi $\CKH$~(\rocqdoc{GHC.CKH.html\#CKH_prv}) and $\WKH$~(\rocqdoc{GHC.CKH.html\#WKH_prv}).}
\label{fig:Hilbert}
\end{figure}

\section{Sequent calculi for $\CK$ and $\WK$}\label{sec:calc}
In this section we introduce the sequent calculi $\GfourCK$~(\rocqdoc{Sequent.G4CK.CKSequents.html\#Provable}) and $\GfourWK$~(\rocqdoc{Sequent.G4WK.WKSequents.html\#Provable}) for $\CK$ and $\WK$, respectively. Sequents in these calculi, which are presented in Figure~\ref{fig:iseq-pc},
are expressions of the shape $\Gamma\Ra\Delta$
where the \emph{antecedent} $\Gamma$ and the \emph{succedent} $\Delta$ are both finite multisets of formulas.
Note that $\Gamma,\phi$ and $\Gamma,\Pi$ stand for $\Gamma\uplus\{\phi\}$ and $\Gamma\uplus\Pi$, respectively.
While antecedents are unconstrained multisets in both calculi,
succedents need to have a cardinality of \emph{exactly one} (i.e.~$|\Delta|=1$) in $\GfourCK$,
and of \emph{at most one} (i.e.~$|\Delta|\leq 1$) in $\GfourWK$.
This technically makes $\GfourCK$ a \emph{single-succedent} calculus,
and $\GfourWK$ a calculus with either \emph{empty} or \emph{singleton} succedents.%
\footnote{In Rocq, instead of using multisets restricted by a cardinality of at most 1,
which is cumbersome, we used the \texttt{option} type%
~(\href{https://rocq-prover.org/doc/V9.1.0/corelib/Corelib.Init.Datatypes.html\#option}{\raisebox{-0.6mm}{\includegraphics[height=0.8em]{orocql.png}}}).
This type takes another type as argument, in our case the type of formulas,
and has two constructor: 
\texttt{None}, simulating~$\emptyset$,
and \texttt{Some($\phi$)}, simulating $\{\phi\}$.}
For simplicity, we call both calculi single-succedent.
Beyond their succedents, these calculi differ on their rule for diamond on the left:
the rule $\DiamL$ of $\GfourCK$ is replaced by the rule $\DiamLW$ in $\GfourWK$.
Note that $\xDiam\Delta$ is non-empty only if $\Delta=\{\Diam\delta\}$.
As shown below, $\DiamLW$ is crucial to prove $\Nax$.
\[
\infer[\scriptstyle\impR]{\Ra\Diam\bot\to\bot}{
    \infer[\scriptstyle\DiamLW]{\Diam\bot\Ra\bot}{
        \infer[\scriptstyle\botL]{\bot\Ra}{}
    }
}
\]
We write $\vdash_{\sys{S}}\Gamma\Ra\Delta$ when $\Gamma\Ra\Delta$ is provable
in $\sys{S}$, with $\sys{S}\in\{\GfourCK,\GfourWK\}$.

\begin{figure}[t!]
\centering
{\small
\hypertarget{rule:botL}{
$\begin{prooftree}
\infer0[$\scriptstyle\botL$]{\Gamma,\bot\Ra\Delta}
\end{prooftree}$}
\quad
\hypertarget{rule:IdP}{
$\begin{prooftree}
\infer0[$\scriptstyle\IdP$]{\Gamma,p\Ra p}
\end{prooftree}$
}
\quad
\hypertarget{rule:andL}{
$\begin{prooftree}
\hypo{\Gamma,\varphi,\psi\Ra\Delta}
\infer1[$\scriptstyle\andL$]{\Gamma,\varphi\land\psi\Ra\Delta}
\end{prooftree}$}
\quad
\hypertarget{rule:andR}{
$\begin{prooftree}
\hypo{\Gamma\Ra\varphi}
\hypo{\Gamma\Ra\psi}
\infer2[$\scriptstyle\andR$]{\Gamma\Ra\varphi\land\psi}
\end{prooftree}$}
\\[0.5cm]
\hypertarget{rule:orL}{
$\begin{prooftree}
\hypo{\Gamma,\varphi\Ra\Delta}
\hypo{\Gamma,\psi\Ra\Delta}
\infer2[$\scriptstyle\orL$]{\Gamma,\varphi\lor\psi\Ra\Delta}
\end{prooftree}$}
\quad
\hypertarget{rule:orR}{
$\begin{prooftree}
\hypo{\Gamma\Ra\varphi_i}
\infer1[$\scriptstyle{\sysrule{$\lor R_i$}(i\in\{1, 2\})} $]{\Gamma\Ra\varphi_1\lor\varphi_2}
\end{prooftree}$}
\hypertarget{rule:impR}{
$\begin{prooftree}
\hypo{\Gamma,\varphi\Ra\psi}
\infer1[$\scriptstyle\impR$]{\Gamma\Ra \varphi\rightarrow\psi}
\end{prooftree}$}
\\[0.5cm]
\hypertarget{rule:impandL}{
$\begin{prooftree}
\hypo{\Gamma,\varphi\rightarrow (\psi\rightarrow\chi)\Ra\Delta}
\infer1[$\scriptstyle\impandL$]{\Gamma,(\varphi\land\psi)\rightarrow\chi\Ra\Delta}
\end{prooftree}$}
\quad
\hypertarget{rule:imporL}{
$\begin{prooftree}
\hypo{\Gamma,\varphi\rightarrow\chi,\psi\rightarrow\chi\Ra\Delta}
\infer1[$\scriptstyle\imporL$]{\Gamma,(\varphi\lor\psi)\rightarrow\chi\Ra\Delta}
\end{prooftree}$}
\\[0.5cm]
\hypertarget{rule:imppL}{
$\begin{prooftree}
\hypo{\Gamma,p,\varphi\Ra\Delta}
\infer1[$\scriptstyle\imppL$]{\Gamma,p,p\rightarrow\varphi\Ra\Delta}
\end{prooftree}$}
\quad
\hypertarget{rule:impimpL}{
$\begin{prooftree}
\hypo{\Gamma,\psi\rightarrow\chi\Ra \varphi\rightarrow\psi}
\hypo{\Gamma,\chi\Ra\Delta}
\infer2[$\scriptstyle\impimpL$]{\Gamma,(\varphi\rightarrow\psi)\rightarrow\chi\Ra\Delta}
\end{prooftree}$}
\\[0.5cm]
\hypertarget{rule:BoxR}{
\begin{prooftree}
\hypo{\openboxes\Gamma \Ra \phi } %
\infer1[$\scriptstyle\BoxR$]{\Gamma \Ra \Box \phi} %
\end{prooftree}}
\quad
\hypertarget{rule:impboxL}{
\begin{prooftree}\hypo{\openboxes\Gamma \Ra \phi}
\hypo{\Gamma,\psi \Ra \Delta }
\infer2[$\scriptstyle\impboxL$]{\Gamma,\Box \phi \rightarrow \psi \Ra \Delta} %
\end{prooftree}}
\\[0.5cm]
\hypertarget{rule:impdiamL}{
\begin{prooftree}\hypo{\openboxes\Gamma,\gamma \Ra \phi}
\hypo{\Gamma,\Diam\gamma , \psi \Ra \Delta}
\infer2[$\scriptstyle\impdiamL$]{\Gamma,\Diam\gamma ,  \Diam \phi \rightarrow \psi \Ra \Delta}
\end{prooftree}}
\\[0.5cm]
\hypertarget{rule:DiamL}{
\begin{prooftree}
\hypo{\openboxes\Gamma,\phi \Ra \psi} %
\infer1[$\scriptstyle\DiamL$]{\Gamma,\Diam\phi \Ra \Diam\psi} %
\end{prooftree}}
\quad
\hypertarget{rule:DiamLW}{
\begin{prooftree}
\hypo{\xBox\Gamma,\phi \Ra \xDiam\Delta} %
\infer1[$\scriptstyle\DiamLW$]{\Gamma,\Diam\phi \Ra \Delta} %
\end{prooftree}}
}

\caption[]{The sequent calculi~$\GfourCK$ and $\GfourWK$.
In the former succedents $\Delta$ are singletons of the shape $\{\delta\}$,
while in the latter succedents are either singletons or the empty set.
The rule $\DiamL$ (resp.~$\DiamLW$) belongs to $\GfourCK$ (resp.~$\GfourWK$) exclusively.}

  \label{fig:iseq-pc}
\end{figure}

$\GfourCK$ and $\GfourWK$ are connected to multiple calculi in the literature.
First, they extend with rules for modalities the strongly terminating calculus $\Gfourip$
for~$\IL$, which has been invented many times~\cite{Vor58,Dyc92,Hud93}.
Second, they are extensions of Iemhoff's single-succedent calculus for~$\iK$~\cite{Iem18}
with rules for~$\Diam$.
Third, our calculi are single-succedent versions of Dalmonte, Grellois and Olivetti's 
multi-succedent calculi for $\CK$ and $\WK$ (which they call $\sys{CCDL}$)~\cite{DalGreOli21}.
As suggested by these last authors, the adaptation of their calculi to a single-succedent setting is straightforward.
The only noteworthy element in this transformation is the use of the single rule $\DiamLW$ in $\GfourWK$ for the treatment of $\Diam$ on the left, in contrast with their counterpart multi-succedent calculus which has two.

\subsection{Termination, cut admissibility and equivalence}\label{sec:props}

In this section we prove of both our calculi that they strongly terminate, eliminate cut, and are equivalent to their corresponding axiomatic system.

First, we prove that $\GfourCK$ and $\GfourWK$ \emph{strongly terminate}.
We say that a calculus strongly terminates if the process of naive backward proof-search,
i.e.~freely and iteratively applying rules backward on a sequent,
necessarily comes to a halt.
To show strong termination, it is sufficient to provide a well-founded ordering
on sequents which decreases upwards in any application of any
rule of our calculi.
To define such an order we exploit the Dershowitz-Manna order on finite multisets~\cite{DerMan}:
a finite multiset $A$ of elements of type $T$ is smaller than another finite multiset $B$ of $T$-elements 
if $A$ is obtained from $B$ by replacing $T$-elements from $B$ by finitely many
$T$-elements which are strictly smaller according to a well-founded order over $T$.

\begin{definition}[Sequent ordering \rocqdoc{Sequent.G4CK.CKOrder.html\#pointed_env_order},\rocqdoc{Sequent.G4WK.WKOrder.html\#pointed_env_order}]
We define the well-founded order $\prec_f$ on formulas using their weight:
$\phi\prec_f\psi$ whenever $w(\phi)<w(\psi)$.
We generate the well-founded Dershowitz-Manna order on multisets $\prec_m$ using $\prec_f$.
Finally, we write $\Gamma \Ra \Delta \prec \Pi \Ra \Sigma$ whenever
$\Gamma\uplus\Delta\prec_m\Pi\uplus\Sigma$.
\end{definition}

Note that a $\GfourWK$ sequent $\Gamma\Ra$ with an empty succedent
is compared in the order over sequents via the multiset $\Gamma\uplus\emptyset = \Gamma$.
A quick inspection of the rules of our calculi shows that any 
premise of a rule is smaller in $\prec$ than its conclusion. 

\begin{proposition}
The calculi $\GfourCK$ and $\GfourWK$ strongly terminate.
\end{proposition}

This directly establishes the decidability of provability in $\GfourCK$ and $\GfourWK$.

\begin{proposition}[\rocqdoc{Sequent.G4CK.CKDecisionProcedure.html\#Provable_dec},\rocqdoc{Sequent.G4WK.WKDecisionProcedure.html\#Provable_dec}]
Provability in $\GfourCK$ and $\GfourWK$ is decidable.
\end{proposition}

In Theorem~\ref{thm:equiv}, at the end of this section, we prove that our calculi capture $\CK$ and $\WK$, respectively.
In this light, the proposition above constitutes a \emph{constructive and mechanised} proof of
decidability for these logics.
Such proofs for the two logics were already given by Dalmonte et al.~\cite{DalGreOli21}, 
constructively but not mechanised,
and for $\CK$ in particular by Mendler and de Paiva~\cite{MenPai05},
though neither constructive nor mechanised.

Second, we aim at proving cut elimination for both calculi,
building on proofs of cut admissibility via local proof transformations.
To get there, we need to follow a path,
first established by Dyckhoff and Negri~\cite{DycNeg00},
paved by a succession of technical lemmas, and
leading to the admissibility of contraction.
We omit the majority of these intermediate lemmas,
referring to our mechanisation and Dyckhoff and Negri's work~\cite{DycNeg00},
except for the admissibility of some useful rules. 

\begin{lemma}[\rocqdoc{Sequent.G4CK.CKSequentProps.html\#generalised_weakeningR},\rocqdoc{Sequent.G4WK.WKSequentProps.html\#generalised_weakeningR},%
\rocqdoc{Sequent.G4CK.CKSequentProps.html\#generalised_axiom},\rocqdoc{Sequent.G4WK.WKSequentProps.html\#generalised_axiom},%
\rocqdoc{Sequent.G4CK.CKSequentProps.html\#weak_ImpL},\rocqdoc{Sequent.G4WK.WKSequentProps.html\#weak_ImpL}]
\label{lemma:weakening}
The following rules are admissible in $\GfourCK$ and $\GfourWK$.
\begin{center}
\begin{tabular}{c @{\hspace{1cm}} c @{\hspace{1cm}} c}
\hypertarget{rule:wk}{
$\infer[\scriptstyle\wk]{\Gamma,\Pi\Ra\Delta}{\Gamma\Ra\Delta}$
}
&
\hypertarget{rule:Id}{$\infer[\scriptstyle\Id]{\Gamma,\phi\Ra\phi}{}$}
&
\hypertarget{rule:impL}{
$
\infer[\scriptstyle\impL]{\Gamma,\phi\to\psi\Ra\Delta}{
    \Gamma\Ra\phi
    &
    \psi,\Gamma\Ra\Delta
}
$
}
\\
\end{tabular}
\end{center}
\end{lemma}

Next, we state the admissibility of contraction.

\begin{proposition}[\rocqdoc{Sequent.G4CK.CKSequentProps.html\#contraction},\rocqdoc{Sequent.G4WK.WKSequentProps.html\#contraction}]
The contraction rule is admissible in $\GfourCK$ and $\GfourWK$.
$$
\hypertarget{rule:cntr}{
\infer[\scriptstyle\cntr]{\Gamma,\psi\Ra\Delta}{\Gamma,\psi,\psi\Ra\Delta}
}
$$
\end{proposition}

We leverage contraction to prove the admissibility of the cut rule.
Our proof goes by primary induction on the weight of the cut formula
and secondary (well-founded) induction on $\prec$,
a standard method for terminating calculi~\cite{GorRamShi21,GorShi22,Shi23,ShiGieGorIem23}.

\begin{theorem}[\rocqdoc{Sequent.G4CK.CKCut.html\#additive_cut},\rocqdoc{Sequent.G4WK.WKCut.html\#additive_cut}]
The additive cut rule is admissible in $\GfourCK$ and $\GfourWK$.
$$
\hypertarget{rule:cut}{
\infer[\scriptstyle\cut]{\Gamma\Ra\Delta}{
    \Gamma\Ra\phi
    &
    \phi,\Gamma\Ra\Delta
    }
}
$$
\end{theorem}

As this theorem is proved syntactically via local proof transformations,
it entails that both calculi \emph{eliminate} cut: 
any proof in the calculus augmented with the cut rule
can be transformed into a proof without instances of cut.

Third, and finally, we show that $\GfourCK$ and $\GfourWK$
capture exactly the logics $\CK$ and $\WK$.
For this, we exploit both the decidability of the sequent
calculi and the admissibility of cut.

\begin{theorem}[\rocqdoc{Sequent.G4CK.CKSoundComplete.html\#G4CK_sound_CKH},\rocqdoc{Sequent.G4CK.CKSoundComplete.html\#G4CK_compl_CKH},%
\rocqdoc{Sequent.G4WK.WKSoundComplete.html\#G4WK_sound_WKH},\rocqdoc{Sequent.G4WK.WKSoundComplete.html\#G4WK_compl_WKH}]\label{thm:equiv}
The following equivalences hold.
\begin{center}
\begin{tabular}{c @{\hspace{0.5cm}} c @{\hspace{0.5cm}} c}
   $\vdash_{\GfourCK}\Gamma \Ra \Delta$ & if and only if & $\Gamma \vdash_{\CK} \Delta$ \\
   $\vdash_{\GfourWK}\Gamma \Ra \Delta$ & if and only if & $\Gamma \vdash_{\WK} \Delta$
\end{tabular}
\end{center}
\end{theorem}

\begin{proof}
We focus on $\GfourWK$, as the equivalence for $\GfourCK$ is treated similarly.
Note that we abused notation: $\Delta$ in $\Gamma \vdash_{\WK} \Delta$
denotes the \emph{formula} $\phi$ if the succedent $\Delta=\{\phi\}$,
and $\bot$ if $\Delta=\emptyset$. 

For the left to right direction it suffices to show all rules of $\GfourWK$ are admissible in $\WK$.

For the right to left direction,
the pen-and-paper proof boils down to showing all axioms provable
and all rules of the axiomatic system admissible in $\GfourWK$.
The admissibility of cut crucially helps in the case of $\MP$.
In the mechanisation, we encounter the following obstacle.
We are trying to prove that $\vdash_{\GfourWK}\Gamma \Ra \Delta$, 
a statement we mechanised in Rocq in the type \texttt{Type},
which requires the construction of an \emph{explicit} proof.
As an assumption we have $\Gamma \vdash_{\WK} \Delta$,
a statement in the type \texttt{Prop}, which promises the 
opaque \emph{existence} of a proof without giving an explicit one.
We are then in a dead end: we cannot build an explicit proof
for $\vdash_{\GfourWK}\Gamma \Ra \Delta$ by relying on the opaque one
of $\Gamma \vdash_{\WK} \Delta$.
To circumvent this difficulty we exploit the decidability of $\GfourWK$:
if the decision procedure outputs $\vdash_{\GfourWK}\Gamma \Ra \Delta$ then we are done,
else we show that the output $\not\vdash_{\GfourWK}\Gamma \Ra \Delta$ and 
our assumption $\Gamma \vdash_{\WK} \Delta$ leads to a contradiction
(which does not require the construction of an explicit proof).
\end{proof}

\section{Uniform interpolation}
We start this section by introducing general definitions and facts about uniform interpolation.

\begin{definition} 
\label{def:UIP}
A modal logic~$L$ over the language $\lang$ has the \emph{uniform interpolation property} if, for every
$\lang$-formula~$\phi$  and variable~$p$, there exist 
$\lang$-formulas, denoted by $\forall p \phi$ and $\exists p \phi$,
satisfying the following three properties:
\begin{enumerate}
\item \emph{$p$-freeness:} \label{UIP:1} $ \pvf{\exists p \phi} \subseteq  \pvf{\phi} \setminus \{ p \}$ and $ \pvf{\forall p \phi} \subseteq  \pvf{\phi} \setminus \{ p \}$,
\item \emph{implication:} \label{UIP:2} $\vdash_L \phi \to \exists p \phi \text{ and } \vdash_L \forall p \phi \to  \phi,$ and
\item \emph{uniformity:} \label{UIP:3} for each formula~$\psi$ with  $p \notin \pvf{\psi}$: 
\begin{align*}
    \vdash_L \phi \to \psi \ &\text{ implies } \ \vdash_L \exists p \phi \to \psi,\\
    \vdash_L \psi \to \phi \ &\text{ implies } \ \vdash_L \psi \to \forall p \phi.
\end{align*}
\end{enumerate}
Formula $\exists p \phi$ is called the \emph{uniform post-interpolant} of $\phi$ w.r.t.~$p$ and $\forall p \phi$ the \emph{uniform pre-interpolant} of $\phi$ w.r.t.~$p$.
\end{definition}

The notation of $\exists p \phi$ and $\forall p \phi$ is suggestive: the formulas do not really contain quantifiers, but they are defined over the modal language $\lang$. This notation is justified by Pitts' result stating that uniform interpolants provide \emph{an interpretation} for propositional quantifiers of second order intuitionistic logic into the propositional language \cite{Pit92}.

\begin{remark}
It is interesting to mention that both in classical and intuitionistic based (modal) logics, the formulas $\forall p (\phi \to \psi)$ and $\exists p (\phi) \to \forall p (\phi \to \psi)$ are equivalent. This was shown in \cite[Lemma 1]{FerGieGooShi24} for modal logics based on $\Box$. Since the proof only relies on intuitionistic propositional reasoning, the same result holds in the intuitionistic modal setting with $\Box$ and $\Diam$. So in case uniform interpolants exist for $\CK$ (which we will show in this paper), we have $\vdash_\CK \forall p (\phi \to \psi)$ if and only if $\vdash_\CK \exists p (\phi) \to \forall p (\phi \to \psi)$. The analogous result holds for $\WK$.
\end{remark}

\begin{remark}
\label{remark:dual}
    In a classical setting, $\forall p \phi$ and $\exists p \phi$ are dual to each other, but that is not true when working in an intuitionistic logic like $\CK$. However, note that it is possible to define formulas of the form $\exists p \phi$ in terms of $\forall$ and $\to$ as follows using an extra propositional variable $q$ not free in $\phi$:
\[
\exists p \phi := \forall q (\forall p (\phi \to q) \to q).
\]
This is a folklore result when the quantifiers denote the quantifiers in second order intuitionistic propositional logic (see \cite{Pit92}). Here we write the quantifiers to denote uniform interpolants for which the equivalence also holds (\cite{Gie22}, Remark~2.2.7). 
This means that from a method constructing all formulas of the form $\forall p \phi$, one gets a definition of $\exists p \phi$. 
However, the known proof-theoretic method \cite{Pit92} for~$\IL$ constructs $\forall p \phi$ and $\exists p \phi$ via a mutual recursion. 
\end{remark}

We provide a proof-theoretic construction to prove the uniform interpolation property. 
In light of Remark~\ref{remark:dual} we use a sequent-style definition of uniform interpolation in which pre- and post-uniform interpolants are constructed simultaneously. 

\begin{definition}
\label{def SUIP}
A set of provable single-succedent 
sequents, denoted~$\vdash$, has the \emph{uniform interpolation property} if, for any single-succedent 
sequent $\Gamma \To \Delta$ and variable $p$, there exist modal formulas $\Ep{\Gamma}$ and $\Ap{\Gamma \To \Delta}$ such
that the following three properties hold:
\begin{enumerate}
    \item \emph{$p$-freeness:}
    \begin{enumerate}[(a)]
        \item \label{p-free_a}$\pvf{\Ep{\Gamma}} \subseteq \pvs{\Gamma} \setminus \{ p \}$, and
        \item \label{p-free_b} $\pvf{\Ap{\Gamma \To \Delta}} \subseteq \pvs{\Gamma, \Delta} \setminus \{ p \}$;
    \end{enumerate}
    \item \emph{implication:}
    \begin{enumerate}[(a)]
        \item \label{implication_a} $\vdash \Gamma \To \Ep{\Gamma}$, and
        \item \label{implication_b} $\vdash \Gamma, \Ap{\Gamma \To \Delta} \To \Delta$;
    \end{enumerate}
    \item \emph{uniformity:} for any finite multiset of formulas $\Pi$ such that $p \notin \pvs{\Pi}$, if it holds that
        $\vdash \Pi, \Gamma \To \Delta$, then it also holds that:  
        \begin{enumerate}
       \item \label{uniformity_a} $\vdash \Pi, \Ep{\Gamma} \To \Delta$ if  $p \notin \pvf{\Delta}$, and
       \item \label{uniformity_b} $\vdash \Pi, \Ep{\Gamma} \To \Ap{\Gamma \To \Delta}$.
       \end{enumerate}
\end{enumerate}
We say that a sequent calculus $\sf{S}$ has the \emph{uniform interpolation property} if $\vdash_{\sf{S}}$ has the uniform interpolation property.
\end{definition}

Note that in the above definition, we have $|\Delta| = 1$ for $\GfourCK$ and $|\Delta| \leq 1$ for~$\GfourWK$. In the latter case, Definition~\ref{def SUIP} is equivalent to the uniform interpolant properties from \cite{Iemhoff19b}. We have the following well-known fact.

\begin{lemma}
    Suppose sequent calculus $\sf{S}$ is sound and complete with respect to logic $L$. If $\sf{S}$ has the uniform interpolation property, then $L$ has the uniform interpolation property.
\end{lemma}
\begin{proof}
    The proof is well-known and can be found in \cite{Bil06,Pit92}. The essence is to define $\exists p \phi := \Ep{\phi}$ and $\forall p \phi := \Ap{\  \To \phi}$.
\end{proof}

\subsection{Uniform interpolation for $\CK$}

To prove the uniform interpolation property for $\CK$, we use the sequent calculus $\GfourCK$ to construct $\Ep{\Gamma}$ and $\Ap{s}$ for any multiset $\Gamma$ and any sequent $s$. Our construction adapts the one for $\iK$ \cite{Iemhoff19b} by adding extra cases for the $\Diam$. 

\begin{figure}[!pht]
\renewcommand*{\arraystretch}{1.15}
\scalebox{0.9}{
\begin{tabular}{| @{\ } l  @{\ }| @{\ \ \ } l @{\ \ \ } | @{\ \ \ }  l @{\ \ \ } |}
\hline
    & $\Gamma$ matches & $\callEp{\Gamma}$ contains\\
\hline
    $(\symEpCK0)$ & $\Gamma',\bot$ & $\bot$ \\
    $(\symEpCK1')$ & $\Gamma',q$ & $q$ \\
    $(\symEpCK1'')$ & $\Gamma',p$ & $\top$\\
    $(\symEpCK2)$ & $\Gamma',\psi_1 \wedge \psi_2$ & $\Ep{\Gamma', \psi_1, \psi_2}$ \\
    $(\symEpCK3)$ & $\Gamma',\psi_1 \vee \psi_2$ & $\Ep{\Gamma',\psi_1} \vee \Ep{\Gamma',\psi_2}$ \\
    $(\symEpCK4')$ & $\Gamma',q \to \psi$ with $q\notin \Gamma'$ & $q \to \Ep{\Gamma',\psi}$ \\
    $(\symEpCK4'')$ & $\Gamma',p \to \psi$ with $p\notin \Gamma'$ & $\top$ \\
    $(\symEpCK5')$ & $\Gamma',r, r \to \psi$ & $\Ep{\Gamma',r,\psi}$\\
    $(\symEpCK6)$ & $\Gamma',(\delta_1 \wedge \delta_2) \to \delta_3$ & $\Ep{\Gamma',\delta_1 \to (\delta_2 \to \delta_3)}$ \\
    $(\symEpCK7)$ & $\Gamma',(\delta_1 \vee \delta_2) \to \delta_3$ & $\Ep{\Gamma',\delta_1 \to \delta_3,\delta_2 \to \delta_3}$ \\    
    $(\symEpCK8')$ & $\Gamma',(\delta_1 \to \delta_2) \to \delta_3$ & \begin{tabular}{@{}r@{}}
    $[\Ep{\Gamma',\delta_2 \to \delta_3,\delta_1} \to \Ap{\Gamma',\delta_2 \to \delta_3,\delta_1 \To \delta_2}]$ \ \\
    $\to \Ep{\Gamma',\delta_3}$
    \end{tabular}
      \\
\hline  
    $(\symEpCK9)$ & $\Gamma \neq \emptyset$ & 
    $\Box \Ep{\xBox \Gamma}$ \\
    $(\symEpCK10)$ & $\Gamma',\Box \delta_1 \rightarrow \delta_2$ & 
    $\Box[\Ep{\xBox\Gamma'}\rightarrow  \Ap{\xBox\Gamma' \Ra \delta_1}]
     \rightarrow  \Ep{\Gamma' ,  \delta_2}$
    \\
\hline  
    $(\symEpCK11)$ & $\Gamma',\Diam\delta$ & 
    $\Diam \Ep{\xBox\Gamma' , \delta}$ \\
    $(\symEpCK12)$ & $\Gamma',\Diam\delta_3,\Diam \delta_1 \rightarrow \delta_2$ & 
    \begin{tabular}{@{}r@{}}
    $\Box[\Ep{\xBox\Gamma',\delta_3}\rightarrow  \Ap{\xBox\Gamma',\delta_3 \Ra \delta_1}]$\\
    $ \rightarrow  \Ep{\Gamma',\Diam\delta_3 ,  \delta_2}$
    \end{tabular} \\
    $(\symEpCK13)$ & $\Gamma',\Diam \delta_1 \rightarrow \delta_2$ & 
    $\Diam[\Ep{\xBox\Gamma'}\rightarrow  \Ap{\xBox\Gamma' \Ra \delta_1}]\rightarrow  \Ep{\Gamma',\delta_2}$
    \\
\hline 
\hline
    & $s$ matches & $\callAp{s}$ contains\\
\hline
    $(\symApCK1')$ & $\Gamma,q \To \phi$ & $\bot$\\
    $(\symApCK1'')$ & $\Gamma,p \To \phi$ with $\phi\neq p$ & $\bot$ \\
    $(\symApCK2)$ & $\Gamma,\psi_1 \wedge \psi_2 \To \phi$ & $\Ap{\Gamma, \psi_1, \psi_2 \To \phi}$ \\
    $(\symApCK3)$ & $\Gamma,\psi_1 \vee \psi_2 \To \phi$ & \begin{tabular}{@{}r@{}}
    $[\Ep{\Gamma,\psi_1} \to \Ap{\Gamma,\psi_1 \To \phi}] \wedge$\\
    $[\Ep{\Gamma,\psi_2} \to \Ap{\Gamma,\psi_2 \To \phi}]$
    \end{tabular}\\
    $(\symApCK4')$ & $\Gamma,q \to \psi \To \phi$ with $q \notin \Gamma$ & $q \wedge \Ap{\Gamma,\psi \To \phi}$ \\
    $(\symApCK4'')$ & $\Gamma,p \to \psi \To \phi$ with $p \notin \Gamma$ & $\bot$ \\
    $(\symApCK5')$ & $\Gamma,r, r \to \psi \To \phi$ & $\Ap{\Gamma,r,\psi \To \phi}$\\
    $(\symApCK6)$ & $\Gamma,(\delta_1 \wedge \delta_2) \to \delta_3 \To \phi$ & $\Ap{\Gamma,\delta_1 \to (\delta_2 \to \delta_3) \To \phi}$ \\
    $(\symApCK7)$ & $\Gamma,(\delta_1 \vee \delta_2) \to \delta_3 \To \phi$ & $\Ap{\Gamma,\delta_1 \to \delta_3,\delta_2 \to \delta_3 \To \phi}$ \\    
    $(\symApCK8')$ & $\Gamma,(\delta_1 \to \delta_2) \to \delta_3 \To \phi$ \ & 
    \begin{tabular}{@{}r@{}}
    $[\Ep{\Gamma,\delta_2 \to \delta_3,\delta_1} \to \Ap{\Gamma,\delta_2 \to \delta_3,\delta_1 \To \delta_2}]$\\
    $\wedge~\Ap{\Gamma,\delta_3 \To \phi}$
    \end{tabular}\\
    $(\symApCK9)$ & $\Gamma \To q$ & $q$\\ 
    $(\symApCK10)$ & $\Gamma, p \To p$ & $\top$\\   
    $(\symApCK11)$ & $\Gamma \To \phi_1 \wedge \phi_2$ & $\Ap{\Gamma \To \phi_1} \wedge \Ap{\Gamma \To \phi_2}$\\
    $(\symApCK12)$ & $\Gamma \To \phi_1 \vee \phi_2$ & $\Ap{\Gamma \To \phi_1} \vee \Ap{\Gamma \To \phi_2}$  \\    
    $(\symApCK13)$ & $\Gamma \To \phi_1 \to \phi_2$ & $\Ep{\Gamma,\phi_1} \to \Ap{\Gamma, \phi_1 \To \phi_2}$\\
\hline 
    $(\symApCK14)$ & $\Gamma \To \Box \delta$ & $\Box(\Ep{\xBox\Gamma} \rightarrow \Ap{\xBox\Gamma\Ra \delta})$\\ 
    $(\symApCK15)$ & $\Gamma , \Box\delta_1 \rightarrow \delta_2 \Ra \phi $ & 
    $\Box[\Ep{\xBox\Gamma}\rightarrow \Ap{\xBox\Gamma\Ra \delta_1}] \wedge \Ap{\Gamma , \delta_2 \Ra \phi}$\\
\hline
    $(\symApCK16)$ & $\Gamma, \Diam \delta_1\To \Diam \delta_2$ & $\Box(\Ep{\xBox\Gamma,\delta_1} \rightarrow \Ap{\xBox\Gamma,\delta_1\Ra \delta_2})$\\ 
    $(\symApCK17)$ & $\Gamma,\Diam\delta_3, \Diam\delta_1 \rightarrow \delta_2 \Ra \phi $ & 
    \begin{tabular}{@{}r@{}}
    $\Box[\Ep{\xBox\Gamma,\delta_3}\rightarrow \Ap{\xBox\Gamma,\delta_3\Ra \delta_1}]$\\
    $ \wedge \Ap{\Gamma,\Diam\delta_3, \delta_2 \Ra \phi}$
    \end{tabular}\\
    $(\symApCK18)$ & $\Gamma\To \Diam \delta$ & $\Diam(\Ep{\xBox\Gamma} \rightarrow \Ap{\xBox\Gamma\Ra \delta})$\\ 
    $(\symApCK19)$ & $\Gamma,\Diam\delta_1 \rightarrow \delta_2 \Ra \phi $ & 
    $\Diam[\Ep{\xBox\Gamma}\rightarrow \Ap{\xBox\Gamma\Ra \delta_1}]\wedge \Ap{\Gamma,\delta_2 \Ra \phi}$ \\
\hline
\end{tabular}
}
\caption{
Constructions of $\Ep{\Gamma}$ and $\Ap{\Gamma \To \phi}$ for $\GfourCK$ where in all clauses $q \neq p$. The constructions are divided into three parts: rows for propositional logic $\IL$ that
are a slight modification from~\cite{Pit92} (indicated here with $'$ or $''$), rows that cover the box modality as in \cite{Iemhoff19b} for box-only logic $\iK$, and rows that cover the diamond. 
}
\label{fig:UIP_IPC}
\end{figure}

\begin{definition}
\label{def_Ap_Ep}
We define sets $\callEp{\Gamma}$~\emph{(\rocqdoc{Sequent.G4CK.CKPropQuantifiers.html\#e_rule},\rocqdoc{Sequent.G4CK.CKPropQuantifiers.html\#e_rule_9},\rocqdoc{Sequent.G4CK.CKPropQuantifiers.html\#e_rule_12})}
and $\callAp{s}$~\emph{(\rocqdoc{Sequent.G4CK.CKPropQuantifiers.html\#a_rule_env},\rocqdoc{Sequent.G4CK.CKPropQuantifiers.html\#a_rule_env_form16},\rocqdoc{Sequent.G4CK.CKPropQuantifiers.html\#a_rule_env17},\rocqdoc{Sequent.G4CK.CKPropQuantifiers.html\#a_rule_form})}
based on a mutual recursion on the ordering $\prec$ according to the rows in Figure~\ref{fig:UIP_IPC}. We define \emph{(\rocqdoc{Sequent.G4CK.CKPropQuantifiers.html\#EA})}:\vspace{-0.5em}
\begin{align*}
    \Ep{\Gamma} := \bigwedge \callEp{\Gamma} \text{ and } \Ap{s} := \bigvee \callAp{s}.
\end{align*}
\end{definition}

\begin{theorem}
\label{thm:terminating}
    The algorithm defining $\Ep{\Gamma}$ and $\Ap{\Gamma \To \phi}$ is terminating.
\end{theorem}
\begin{proof}
     $\Ep{\Gamma}$ and $\Ap{\Gamma \To \phi}$ are defined by mutual induction on the ordering~$\prec$ where for $\Ap{\Gamma \To \phi}$ we look at the multiset $\Gamma,\phi$. Each recursive call in Figure~\ref{fig:UIP_IPC} reduces in that ordering. Moreover, for each multiset $\Gamma$ there are finitely many matches of rows $(\symEpCK0)\text{-}(\symEpCK13)$, and similarly so for sequent $\Gamma \To \phi$ and rows $(\symApCK1')\text{-}(\symApCK19)$. This means that $\callAp{\Gamma \To \phi}$ and $\callEp{\Gamma}$ are finite sets of formulas. So, the constructions of $\Ep{\Gamma}$ and $\Ap{\Gamma \To \phi}$ are terminating.
\end{proof}

Let us explain the construction of Figure~\ref{fig:UIP_IPC}. Rows $(\symEpCK0),(\symEpCK1'),(\symEpCK4')$ and $(\symApCK1'),(\symApCK4'),(\symApCK9)$ take care of boolean constants and variables other than $p$, which can be smartly taken out of the recursive call since these are atomic $p$-free formulas.
Rows $(\symEpCK1''),(\symEpCK4'')$ and rows $(\symApCK1''),(\symApCK4'')$
are not necessary for the construction, as $\bot$ is a neutral element for $\Ep{\cdot}$ and $\top$ is a neutral element for $\Ap{\cdot}$, but are required in Rocq to ensure the functionality of $\Ep{\cdot}$ and $\Ap{\cdot}$ by making them defined on all inputs.

Rows $(\symEpCK2),(\symEpCK3),(\symEpCK5')\text{-}(\symEpCK8'),(\symEpCK10),(\symEpCK12)$ and $(\symApCK2)$, $(\symApCK3)$, $(\symApCK5')\text{-}(\symApCK8')$, $(\symApCK11)\text{-}(\symApCK17)$ correspond to what we call a \emph{full} application of a calculus rule where the middle column presents the conclusion of the rule and the right column uses the premise of that rule in the recursive call. 

Rows $(\symEpCK9),(\symEpCK11),(\symEpCK13),(\symApCK18)$ and $(\symApCK19)$ correspond to what we call \emph{partial} applications of a rule from the calculus. When considering the \emph{uniformity} property of uniform interpolation in Definition~\ref{def SUIP}, a rule application might be possible due to extra formulas in the $p$-free context $\Pi$ outside the reach of the interpolant construction meaning that these rows in the table \textit{partially} match the form of rule from the calculus. For example, row $(\symEpCK11)$ is used in case there is a formula $\Diam \delta'_1 \to \delta_2'$ in the hidden $p$-free context that would together with $\Diam \delta \in \Gamma$ make an application of the rule $\impdiamL$ possible. Row $(\symEpCK9)$ is used for partial applications of $\BoxR$ and $\DiamL$. 

We will use the following simple facts.

\begin{lemma}
\label{Lemma_simple1}
\phantom{A}
\begin{enumerate}
    \item\label{bigwedgeE} Let $\Gamma$ be a multiset such that $\callEp{\Gamma} \neq \emptyset$. For any $\chi \in \callEp{\Gamma}$, multiset~$\Pi$ and formula $\phi$, if $\vdash_\GfourCK \Pi,\chi \To \phi$, then $\vdash_\GfourCK \Pi,\Ep{\Gamma} \To \phi$.%
    ~\emph{(\rocqdoc{Sequent.G4CK.CKPropQuantifiers.html\#E_left},\rocqdoc{Sequent.G4CK.CKPropQuantifiers.html\#E9_left},\rocqdoc{Sequent.G4CK.CKPropQuantifiers.html\#E12_left})}
    \item\label{bigveeA} Let $\Gamma \To \phi$ be a sequent such that $\callAp{\Gamma \To \phi}\neq\emptyset$. For any $\chi \in \callAp{\Gamma \To \phi}$ and multiset~$\Pi$, if $\vdash_\GfourCK \Pi \To \chi$, then $\vdash_\GfourCK \Pi \To \Ap{\Gamma \To \phi}$.%
    ~\emph{(\rocqdoc{Sequent.G4CK.CKPropQuantifiers.html\#A_right},\rocqdoc{Sequent.G4CK.CKPropQuantifiers.html\#Af_right},\rocqdoc{Sequent.G4CK.CKPropQuantifiers.html\#A16_right},\rocqdoc{Sequent.G4CK.CKPropQuantifiers.html\#A17_right})}
\end{enumerate}
\end{lemma}
\begin{proof}
    \eqref{bigwedgeE} follows from the fact that $\Ep{\Gamma} = \bigwedge \callEp{\Gamma}$. So multiple applications of weakening and $\andL$ applied to $\Pi,\chi \To \phi$ gives a proof for $\Pi,\Ep{\Gamma} \To \phi$. For~\eqref{bigveeA}, observe that $\Ap{\Gamma\To\phi}=\bigvee\callAp{\Gamma\To\phi}$, so multiple applications of $\orRone$ and $\orRtwo$ applied to $\Pi \To \chi$ provide a proof for $\Pi \To \Ap{\Gamma \To \phi}$.
\end{proof}

\begin{lemma}[\rocqdoc{Sequent.G4CK.CKPropQuantifiers.html\#E9_left}]
\label{Lemma_simple2}
    If $\vdash_\GfourCK \Pi, \Box\Ep{\xBox\Gamma} \To \phi$, then $\vdash_\GfourCK \Pi, \Ep{\Gamma} \To \phi$.
\end{lemma}
\begin{proof}
   If $\Gamma = \emptyset$, then $\xBox\Gamma = \emptyset$ and so $\Box\Ep{\xBox\Gamma}=\Box\top=\top$. In that case $\vdash_\GfourCK \Pi,\top \To \phi$, and so $\vdash_\GfourCK \Pi,\Ep{\Gamma} \To \phi$ by weakening. If $\Gamma \neq \emptyset$ we know by $(\symEpCK9)$ that $\Box\Ep{\xBox\Gamma} \in \callEp{\Gamma}$. So by Lemma~\ref{Lemma_simple1} we conclude that $\vdash_\GfourCK \Pi, \Ep{\Gamma} \To \phi$. 
\end{proof}

We now present the main theorem of this section.
Due to space restrictions we highlight some of the important steps of the proof below. 

\begin{theorem}
\label{thm:UIP_CK}
    Sequent calculus $\GfourCK$ has the uniform interpolation property.
\end{theorem}
\begin{proof}
    We take the construction of $\Ap{\cdot}$ and $\Ep{\cdot}$ from Definition~\ref{def_Ap_Ep}. We have to show all properties from Definition~\ref{def SUIP}. The \emph{$p$-freeness} property follows easily by examining the construction of $\Ap{\cdot}$ and $\Ep{\cdot}$~(\rocqdoc{Sequent.G4CK.CKPropQuantifiers.html\#EA_vars}). 

    The \emph{implication} properties \eqref{implication_a} and \eqref{implication_b} are simultaneously proved by the multiset ordering $\prec$ on $\Gamma,\phi$, i.e., we prove \eqref{implication_a} $\vdash_\GfourCK \Gamma \To \Ep{\Gamma}$, and \eqref{implication_b} $\vdash_\GfourCK \Gamma, \Ap{\Gamma \To \phi} \To \phi$.
    The implementation~(\rocqdoc{Sequent.G4CK.CKPropQuantifiers.html\#entail_correct}) is an elegant extension of the existing implementation for $\IL$ \cite{Fer23}. For \eqref{implication_a}, we have $\Ep{\Gamma} := \bigwedge \callEp{\Gamma}$ given in Figure~\ref{fig:UIP_IPC}, so it is sufficient to prove every conjunct in $\callEp{\Gamma}$. Let us here treat two cases for $\Diam$.\vspace{0.5em}\\ 
\noindent \emph{Case $(\symEpCK9)$}:\\ In this case $\Gamma \neq \emptyset$, which means that $\xBox \Gamma \prec \Gamma$, so the induction hypothesis~(IH) applies: \vspace{-0.5em}
    \begin{center}
    \small
        \begin{prooftree}      \hypo{\scriptstyle{\text{(IH)}}} %
        \infer[no rule]1[]{\xBox \Gamma \To \Ep{\xBox \Gamma}} %
        \infer1[$\scriptstyle\BoxR$]{\Gamma \To \Box \Ep{\xBox \Gamma}}
    \end{prooftree}
    \end{center}
    \vspace{-0.5em}
    \noindent \emph{Case $(\symEpCK13)$}:\\ Let us write $\chi := \Ep{\xBox \Gamma'}\to\Ap{\xBox\Gamma'\To\delta_1}$. We have the following derivation in which we use the rules $\impL$ and $\wk$ shown admissible in Lemma~\ref{lemma:weakening}.\vspace{-0.5em}
    \begin{center}
    \scalebox{0.85}{
        \begin{prooftree}
        \hypo{\scriptstyle{\text{(IH)}}} %
        \infer[no rule]1[]{\xBox\Gamma' \To \Ep{\xBox\Gamma'}} %
        \hypo{\scriptstyle{\text{(IH)}}} %
        \infer[no rule]1[]{\xBox\Gamma',\Ap{\xBox\Gamma'\To\delta_1} \To \delta_1} %
        \infer2[$\scriptstyle\impL$]{\xBox\Gamma',\chi \To \delta_1}
        \hypo{\scriptstyle{\text{(IH)}}} %
        \infer[no rule]1[]{\Gamma',\delta_2 \To \Ep{\Gamma', \delta_2}} %
        \infer1[$\scriptstyle\wk$]{\Gamma', \delta_2, \Diam\chi \To \Ep{\Gamma',\delta_2}} %
        \infer2[$\scriptstyle\impdiamL$]{\Gamma', \Diam\delta_1\to\delta_2,\Diam\chi \To \Ep{\Gamma',\delta_2}} %
        \infer1[$\scriptstyle\impR$]{\Gamma', \Diam\delta_1\to\delta_2 \To \Diam\chi\to\Ep{\Gamma',\delta_2}} %
    \end{prooftree}
    }
    \vspace{0.2em}
    \end{center}

For \eqref{implication_b} we prove that $\vdash_\GfourCK \Gamma, \Ap{\Gamma \To \phi} \To \phi$. We have $\Ap{\Gamma\To\phi} := \bigvee \callAp{\Gamma\To\phi}$ given in Figure~\ref{fig:UIP_IPC}, so it is sufficient to prove every disjunct in $\Ap{\Gamma\To\phi}$ placed on the left of the sequent. 
As an example, we treat the following case.  \vspace{-0.5em}\\

    \noindent \emph{Case $(\symApCK19)$}:\\ Let us write $\chi := \Ep{\xBox \Gamma'}\to\Ap{\xBox\Gamma'\To\delta_1}$. Note that in the following derivation, $\xBox\Ap{\Gamma',\delta_2\To\phi}$ represents a formula when $\Ap{\Gamma',\delta_2\To\phi}$ is a boxed formula and $\emptyset$ when $\Ap{\Gamma',\delta_2\To\phi}$ is not a boxed formula. In both cases, our weakening rule from Lemma~\ref{lemma:weakening} works.
    \begin{center}
    \scalebox{0.8}{ 
        \begin{prooftree}
        \hypo{\scriptstyle{\text{(IH)}}} %
        \infer[no rule]1[]{\xBox\Gamma' \To \Ep{\xBox\Gamma'}} %
        \hypo{\scriptstyle{\text{(IH)}}} %
        \infer[no rule]1[]{\xBox\Gamma',\Ap{\xBox\Gamma'\To\delta_1} \To \delta_1} %
        \infer2[$\scriptstyle\impL$]
        {\xBox\Gamma',\chi \To \delta_1} %
        \infer1[$\scriptstyle\wk$]{\xBox\Gamma',\chi,\xBox\Ap{\Gamma',\delta_2\To\phi} \To \delta_1} %
        \hypo{\scriptstyle{\text{(IH)}}} %
        \infer[no rule]1[]{\Gamma',\delta_2,\Ap{\Gamma',\delta_2\To\phi} \To \phi} %
        \infer1[$\scriptstyle\wk$]{\Gamma',\delta_2,\Diam\chi,\Ap{\Gamma',\delta_2\To\phi} \To \phi} %
        \infer2[$\scriptstyle\impdiamL$]{\Gamma',\Diam\delta_1\to\delta_2,\Diam\chi,\Ap{\Gamma',\delta_2\To\phi} \To \phi} %
        \infer1[$\scriptstyle\andL$]{\Gamma',\Diam\delta_1\to\delta_2,\Diam\chi\wedge\Ap{\Gamma',\delta_2\To\phi} \To \phi} %
    \end{prooftree}
    }
    \vspace{0.5em}
    \end{center}
    
    Now we turn to the \emph{uniformity} property of uniform interpolation~(\rocqdoc{Sequent.G4CK.CKPropQuantifiers.html\#pq_correct}). Let~$\Pi$ be a finite multiset such that $p \notin \pvs{\Pi}$ and suppose $\vdash_\GfourCK \Pi, \Gamma \To \phi$ by proof~$\pi$. By induction on $\Pi, \Gamma \To \phi$ via $\prec$ 
    we simultaneously prove that \eqref{uniformity_a} $\vdash \Pi, \Ep{\Gamma} \To \phi$ if $p \notin \pvf{\phi}$ and \eqref{uniformity_b} $\vdash_\GfourCK \Pi,\Ep{\Gamma} \To \Ap{\Gamma\To\phi}$. In the sequel we use primes to denote that a formula is $p$-free, so for \eqref{uniformity_a} we write $\phi = \phi'$. We look at the last rule applied in proof~$\pi$.
    Here we only discuss the rule $\DiamL$. \vspace{1em}\\
    \noindent \emph{Case $\DiamL$}:\\ \emph{Subcase}: The principal formula is in $p$-free $\Pi$, i.e., $\Pi=\Pi_0,\Diam\delta'_1$.
    \begin{enumerate}[(a)]
        \item The conclusion of the rule is of the form $\Pi_0, \Diam \delta'_1,\Gamma \Ra \Diam \phi_0'$ with $p$-free formula $\phi_0'$. By induction we have $\vdash_\GfourCK \xBox \Pi_0, \delta'_1,\Ep{\xBox\Gamma} \Ra \phi'_0$. We obtain the desired result with the following derivation:
        \begin{center}
        \small
        \begin{prooftree}
        \hypo{\scriptstyle{\text{(IH)}}}
        \infer[no rule]1{\xBox\Pi_0,\delta'_1,\Ep{\xBox\Gamma} \To \phi'_0} %
        \infer1[$\scriptstyle\DiamL$]{\Pi_0,\Diam\delta'_1,\Box\Ep{\xBox\Gamma}\To\Diam\phi'_0}
        \infer[double]1[Lemma~\ref{Lemma_simple2}]{\Pi_0,\Diam\delta'_1,\Ep{\Gamma} \To \Diam\phi_0'}
        \end{prooftree}
        \end{center}
          
        \item The conclusion of the rule is of the form $\Pi_0, \Diam \delta'_1,\Gamma \Ra \Diam \phi$. By induction, we get $\vdash_\GfourCK \xBox \Pi_0, \delta'_1,\Ep{\xBox\Gamma} \Ra \Ap{\xBox\Gamma \Ra \phi}$. We have the following:

        \begin{center}
        \small
        \begin{prooftree}
        \hypo{\scriptstyle{\text{(IH)}}}
        \infer[no rule]1{\xBox\Pi_0,\delta'_1,\Ep{\xBox\Gamma} \To \Ap{\xBox\Gamma\To\phi}} %
        \infer1[$\scriptstyle\impR$]{\xBox\Pi_0,\delta'_1 \To \Ep{\xBox\Gamma}\to\Ap{\xBox\Gamma\To\phi}} %
        \infer1[$\scriptstyle\DiamL$]{\Pi_0,\Diam\delta'_1 \To \Diam[\Ep{\xBox\Gamma}\to\Ap{\xBox\Gamma\To\phi}]} %
        \infer1[$\scriptstyle\wk$]{\Pi_0,\Diam\delta'_1,\Ep{\Gamma} \To \Diam[\Ep{\xBox\Gamma}\to\Ap{\xBox\Gamma\To\phi}]} %
        \infer[double]1[\small{$(\symApCK18)$, Lemma~\ref{Lemma_simple1}}]{\Pi_0,\Diam\delta'_1,\Ep{\Gamma} \To \Ap{\Gamma\To\Diam\phi}}
        \end{prooftree}
        \end{center}
    \end{enumerate}
    \emph{Subcase}: The principal formula is in $\Gamma$, i.e., $\Gamma = \Gamma_0,\Diam\delta_1$.
    \begin{enumerate}[(a)]
        \item This case proceeds in a similar way as the previous subcase (a), where now, instead of Lemma~\ref{Lemma_simple2} we use Lemma~\ref{Lemma_simple1} with $(\symEpCK11)$.
        
         \item This case proceeds in a similar way as the previous subcase (b), where, instead of $(\symApCK18)$ we use $(\symApCK16)$ and rule $\BoxR$ instead of $\DiamL$.
    \end{enumerate}

We have checked the properties of $p$-freeness, implication, and uniformity for sequent calculus $\GfourCK$, concluding that $\GfourCK$ has uniform interpolation.
\end{proof}

\begin{corollary}[\rocqdoc{Sequent.G4CK.CKPropQuantifiers.html\#CK_uniform_interpolation}]
    Logic $\CK$ has the uniform interpolation property.
\end{corollary}

\subsection{Uniform interpolation for $\WK$}

The construction of uniform interpolants for $\WK$ only requires two modifications compared to the construction of uniform interpolants for $\CK$ in Figure~\ref{fig:UIP_IPC}, because the two calculi $\GfourCK$ and $\GfourWK$ only differ in two aspects. First, since sequents $\Gamma \Ra \Delta$ in $\GfourWK$ obey $|\Delta|\leq 1$ instead of $|\Delta|=1$, we should read any row in Figure~\ref{fig:UIP_IPC} that matches a sequent~$s$ with $\phi$ in the succedent, as having $\Delta$ in the succedent with $|\Delta|\leq 1$. So we replace all occurrences of $\phi$ (without a subscript) in Figure~\ref{fig:UIP_IPC} with $\Delta$. To be precise, this applies to rows $(\symApCK1')\text{-}(\symApCK8'), (\symApCK15), (\symApCK17)$, and $(\symApCK19)$.

Second, the main difference between the calculi $\GfourCK$ and $\GfourWK$ lies in the left rule for $\Diam$ which in $\GfourWK$ has the form:
\[
\infer[\scriptstyle\DiamLW]{\Gamma,\Diam\phi \Ra \Delta}{
    \xBox\Gamma,\phi \Ra \xDiam\Delta}
\]
That is why we modify row $(\symApCK16)$ from Figure~\ref{fig:UIP_IPC} to the following row~(\rocqdoc{Sequent.G4WK.WKPropQuantifiers.html\#a_rule_env}):
\begin{center}
\renewcommand*{\arraystretch}{1.5}
\begin{tabular}{| l | l | l |}
\hline
 $\ (\symApWK16)$ \  & \ $\Gamma, \Diam \delta \Ra \Delta$ \ \qquad \quad & \ $\Box (\Ep{\xBox\Gamma,\delta} \to \Ap{\xBox\Gamma, \delta \Ra \xDiam \Delta})$ \quad \ \\
\hline
\end{tabular}
\end{center}

Now, $\Ep{\Gamma}$ and $\Ap{s}$ with respect to $\GfourWK$ are defined as in Definition~\ref{def_Ap_Ep} with the modified table as described above. The analogues of Theorem~\ref{thm:terminating} and Lemmas~\ref{Lemma_simple1} and~\ref{Lemma_simple2} hold for $\GfourWK$ by the same proofs.

\begin{theorem}
\label{thm:UIP_WK}
    Sequent calculus $\GfourWK$ has the uniform interpolation property.
\end{theorem}
\begin{proof}
    (\rocqdoc{Sequent.G4WK.WKPropQuantifiers.html\#EA_vars},\rocqdoc{Sequent.G4WK.WKPropQuantifiers.html\#entail_correct},\rocqdoc{Sequent.G4WK.WKPropQuantifiers.html\#pq_correct}) Compared to the proof of uniform interpolation for $\GfourCK$ (Theorem~\ref{thm:UIP_CK}), a few modifications are needed, mostly in those places where $(\symApCK16)$ plays a role. The proof of the \textit{implication} property \eqref{implication_a} does not change, and for \eqref{implication_b} the case for $(\symApWK16)$ works via the same proof as for $(\symApCK16)$ in Theorem~\ref{thm:UIP_CK}, where we replace each occurrence of $\Diam \delta_2$ by $\Delta$ and each $\delta_2$ by $\xDiam\Delta$. For the \textit{uniformity} property, we check the case for $\DiamLW$ which comes with two subcases. They show similarities to the proof for $\DiamL$ of Theorem~\ref{thm:UIP_CK}:\vspace{0.6em}\\
    \textit{Case} $\DiamLW$:\\
    \textit{Subcase}: The principal formula is in the $p$-free context $\Pi$, i.e., $\Pi = \Pi_0, \Diam\delta'$.
    \begin{enumerate}[(a)]
        \item The conclusion of the rule is of the form $\Pi_0, \Diam\delta',\Gamma \To \Delta'$ where $\Delta'$ is $p$-free. By induction we have $\vdash_\GfourWK \xBox \Pi_0,\delta',\Ep{\xBox\Gamma} \To \xDiam \Delta'$. We have:
        \begin{center}
        \small
        \begin{prooftree}
        \hypo{\scriptstyle{\text{(IH)}}}
        \infer[no rule]1{\xBox\Pi_0,\delta',\Ep{\xBox\Gamma} \To \xDiam\Delta'} %
        \infer1[$\scriptstyle\DiamLW$]{\Pi_0,\Diam\delta',\Box\Ep{\xBox\Gamma}\To\Delta'}
        \infer[double]1[Lemma~\ref{Lemma_simple2}]{\Pi_0,\Diam\delta',\Ep{\Gamma} \To \Delta'}
        \end{prooftree}
        \end{center}
        \item The conclusion of the rule is of the form $\Pi_0, \Diam\delta',\Gamma \To \Delta$. 
        We now consider two possibilities: $\Delta = \{\Diam\psi\}$ for some $\psi$ or not. If not, then $\xDiam\Delta = \emptyset$ and $\Ap{\Gamma \To \Delta}$ is not a $\Diam$-formula by inspection of its construction. This means that we can reason as follows:
        \begin{center}
        \small
        \begin{prooftree}
        \hypo{\scriptstyle{\text{(IH)}}}
        \infer[no rule]1{\xBox\Pi_0,\delta',\Ep{\xBox\Gamma} \To } %
        \infer1[$\scriptstyle\DiamLW$]{\Pi_0,\Diam\delta',\Box\Ep{\xBox\Gamma}\To\Ap{\Gamma \To \Delta}}
        \infer[double]1[Lemma~\ref{Lemma_simple2}]{\Pi_0,\Diam\delta',\Ep{\Gamma} \To \Ap{\Gamma \To \Delta}}
        \end{prooftree}
        \end{center}
        If $\Delta = \{\Diam\psi\}$, then the same derivation applies as case (b) for $\DiamL$ in the proof of Theorem~\ref{thm:UIP_CK}.
    \end{enumerate}
    \textit{Subcase}: The principal formula is in $\Gamma$, i.e., $\Gamma = \Gamma_0, \Diam\delta_1$. This subcase is analogous to the previous subcase where instead of Lemma~\ref{Lemma_simple2} we use Lemma~\ref{Lemma_simple1} with $(\symEpWK11)$ and instead of using $(\symApWK18)$ we apply $(\symApWK16)$.
\end{proof}

\begin{corollary}[\rocqdoc{Sequent.G4WK.WKPropQuantifiers.html\#WK_uniform_interpolation}]
\label{thm:UIP_WK}
    Logic $\WK$ has the uniform interpolation property.
\end{corollary}

\section{Conclusion}
We have established the uniform interpolation property for constructive modal logics $\CK$ and $\WK$, the first positive result of this kind among intuitionistic modal logics with independent $\Box$ and $\Diam$. Our proof is \textit{constructive} where we extend Pitts' proof technique to constructive $\Diam$ by exploiting terminating sequent calculi for $\CK$ and $\WK$. Our results have been formalised in Rocq which could contribute to the uniform interpolation calculator available in \cite{Fer23,FerGieGooShi24}.

Our construction aims at showing that 
uniform interpolants \textit{can} be constructed.
Our approach is based on Pitts' naive construction, which can be simplified and optimised,
as shown for the case of $\IL$ in~\cite{FerGooIgl25}.
Such simplifications, optimisations and hence more efficient algorithms, would be welcome in modal logics.
A recent overview on complexity bounds on (uniform) interpolants in classical modal logics is provided in~\cite{tCate_Kuijer_Wolter25preprint}.

Another interesting direction for future research is to study the definability of so-called \textit{bisimulation quantifiers}. These are interpreted over relational semantics and closely linked to uniform interpolants \cite{Vis1996,dAgostino07}. It would be interesting to see whether these are also definable in the birelational models of intuitionistic modal logics. Birelational semantics have been developed for $\CK$ \cite{MenPai05} and $\WK$ \cite{Wijesekera90} and a general study of such semantics is provided and formalised in Rocq in \cite{GroShiClo25}.

Finally, Pitts' technique received attention in the recent field of \textit{universal proof theory} in which the uniform interpolation is linked to nicely behaved proof systems by inspecting the form of inference rules. So-called negative results have been established in the realm of intermediate logics stating that many of such logics cannot have a \textit{semi-analytic} sequent calculus \cite{Iemhoff19b}. Modal logics are also treated in this context \cite{Iem19a,AkbarTabatabai_Jalali25preprint,AkbarTabatabai_Jalali25}, but the negative results are in our opinion quite restrictive in the modal setting because only a few $\Box$-only rules are treated and a general format of `modal rule' is not provided. We hope to provide a first step to widen the scope of universal proof theory by studying both $\Box$ and $\Diam$.

\begin{credits}
\subsubsection{\ackname} The first author acknowledges the support by the Dutch Research Council (NWO) under the project \textit{Finding Interpolants: Proofs in Action} with file number VI.Veni.232.369 of the research programme Veni.
The second author has been supported by
(1) a UKRI Future Leaders Fellowship, ‘Structure vs Invariants in Proofs’, project reference MR/S035540/1,
and
(2) the Renaissance Philanthropy grant DEEPER.
\end{credits}

\bibliographystyle{splncs04}
\bibliography{mybib}

@PREAMBLE{ {\providecommand{\noopsort}[1]{}} }

@InProceedings{dAgostino07,
author="{\noopsort{Agostino}}{{D}'{A}gostino}, Giovanna",
editor="ten Cate, Balder
and Zeevat, Henk W.",
title="Uniform Interpolation, Bisimulation Quantifiers, and Fixed Points",
booktitle="Logic, Language, and Computation, TbiLLC 2005",
year="2007",
publisher="Springer Berlin Heidelberg",
address="Berlin, Heidelberg",
pages="96--116",
isbn="978-3-540-75144-1",
doi="10.1007/978-3-540-75144-1\_8",
series = "LNCS",
volume = "4363"
}

@misc{AkbarTabatabai_Jalali25preprint,
      title={Universal Proof Theory: Semi-analytic Rules and Uniform Interpolation}, 
      author={Akbar Tabatabai, Amirhossein and Jalali, Raheleh},
      year={2025},
      eprint={1808.06258},
      archivePrefix={arXiv},
      primaryClass={cs.LO},
      url={https://arxiv.org/abs/1808.06258}, 
}

@article{AkbarTabatabai_Jalali25,
title = {{Universal proof theory: Semi-analytic rules and Craig interpolation}},
journal = {Annals of Pure and Applied Logic},
volume = {176},
number = {1},
pages = {103509},
year = {2025},
issn = {0168-0072},
doi = {https://doi.org/10.1016/j.apal.2024.103509},
author = {Akbar Tabatabai, Amirhossein and Jalali, Raheleh}
}

@article{AkbarTabatabai_Iemhoff_Jalali24,
    author = {Akbar Tabatabai, Amirhossein and Iemhoff, Rosalie and Jalali, Raheleh},
    title = {{Uniform Lyndon interpolation for basic non-normal modal and conditional logics}},
    journal = {Journal of Logic and Computation},
    volume = {35},
    number = {6},
    pages = {exae057},
    year = {2024},
    month = {11},
    issn = {0955-792X},
    doi = {10.1093/logcom/exae057}
}

@article{AliDerOno2014, 
title={Uniform interpolation in substructural logics}, 
volume={7}, 
DOI={10.1017/S175502031400015X}, 
number={3}, 
journal={The Review of Symbolic Logic}, 
author={Alizadeh, Majid and Derakhshan, Farzaneh and Ono, Hiroakira},
year={2014}, 
pages={455–483}
}

@inproceedings{AleMenPaiRit01,
  author    = {Alechina, N. and
               Mendler, M. and
               de Paiva, V. and
               Ritter, E.},
  title     = {{Categorical and Kripke semantics for Constructive S4 modal logic}},
  booktitle = {Proceedings {CSL} 2001},
  editor    = {Fribourg, L.},
  OPTseries = {Lecture Notes in Computer Science},
  OPTvolume = {2142}, 
  publisher = {Springer},
  address   = {Berlin, Heidelberg},
  doi       = {10.1007/3-540-44802-0_21},
  year      = {2001},
}

@article{ArisakaDasStrassburger15,
author = {Arisaka, Ryuta and Das, Anupam and Stra{\ss}burger, Lutz},
year = {2015},
number = {3},
title = {On Nested Sequents for Constructive Modal Logics},
volume = {11},
journal = {Logical Methods in Computer Science},
doi = {10.2168/LMCS-11(3:7)2015}
}

@InProceedings{Balbiani_et_al24a,
  author =	{Balbiani, Philippe and Gao, Han and Gencer, {\c{C}}i{\u{g}}dem and Olivetti, Nicola},
  title =	{{A natural intuitionistic modal logic: Axiomatization and bi-nested calculus}},
  booktitle =	{32nd EACSL Annual Conference on Computer Science Logic (CSL 2024)},
  pages =	{13:1--13:21},
  series =	{Leibniz International Proceedings in Informatics (LIPIcs)},
  ISBN =	{978-3-95977-310-2},
  ISSN =	{1868-8969},
  year =	{2024},
  volume =	{288},
  editor =	{Murano, Aniello and Silva, Alexandra},
  publisher =	{Schloss Dagstuhl -- Leibniz-Zentrum f{\"u}r Informatik},
  address =	{Dagstuhl, Germany},
  doi =		{10.4230/LIPIcs.CSL.2024.13}
}

@InProceedings{Balbiani_et_all24b,
author="Balbiani, Philippe
and Gao, Han
and Gencer, {\c{C}}i{\u{g}}dem
and Olivetti, Nicola",
editor="Benzm{\"u}ller, Christoph
and Heule, Marijn J.H.
and Schmidt, Renate A.",
title="Local Intuitionistic Modal Logics and Their Calculi",
booktitle="Automated Reasoning",
year="2024",
publisher="Springer Nature Switzerland",
address="Cham",
pages="78--96"
}

@inproceedings{BeldePRit01,
  title={Extended {C}urry-{H}oward correspondence for a basic constructive modal logic},
  author={Bellin, G. and de Paiva, V. and Ritter, E.},
  booktitle={Methods for Modalities 2 (M4M-2)},
  year={2001}
}

@Article{BiePai00,
  author    = {Bierman, G. M. and
               de Paiva, V.},
  title     = {On an Intuitionistic Modal Logic},
  journal   = {Studia Logica},
  year      = {2000},
  volume    = {65},
  number    = {3},
  pages     = {383-416},
  doi       = {10.1023/A:1005291931660},
  OPTurl    = {https://doi.org/10.1023/A:1005291931660}
}

@phdthesis{Bil06,
author={Marta Bílková},
  title   = {{Interpolation in Modal Logics}},
  school  = {Univerzita Karlova},
  address = {Prague},
  year    = {2006}	
}

@misc{Bil22,
      title={Uniform Interpolation in provability logics}, 
      author={Marta Bílková},
      year={2022},
      eprint={2211.02591},
      archivePrefix={arXiv},
      primaryClass={math.LO},
      url={https://arxiv.org/pdf/2211.02591.pdf}
}

@book{BookCraigInterpolation,
	editor = {Balder {\noopsort{Cate}}{Ten Cate} and Jean Christoph Jung and Patrick Koopmann and Christoph Wernhard and Frank Wolter},
	Publisher = {Ubiquity Press},
	Title = {{Theory and Applications of Craig Interpolation}},
	Year = {To appear in spring 2026}
    }

@article{Bull65,
author = "Bull, R. A.",
doi = "10.1305/ndjfl/1093958154",
fjournal = "Notre Dame Journal of Formal Logic",
journal = "Notre Dame Journal of Formal Logic",
number = "2",
pages = "142--146",
publisher = "Duke University Press",
title = "A modal extension of intuitionistic logic.",
volume = "6",
year = "1965"
}

@article{Bull66,
author = "Bull, R. A.",
fjournal = "Journal of Symbolic Logic",
journal = "The Journal of Symbolic Logic",
number = "4",
pages = "609--616",
publisher = "Association for Symbolic Logic",
title = "{MIPC} as the Formalisation of an Intuitionistic Concept of Modality",
volume = "31",
year = "1966",
doi = "10.2307/2269696"
}

@misc{tCate_Kuijer_Wolter25preprint,
      title={The Size of Interpolants in Modal Logics}, 
      author={Balder {\noopsort{Cate}}{Ten Cate} and Louwe Kuijer and Frank Wolter},
      year={2025},
      eprint={2511.04577},
      archivePrefix={arXiv},
      primaryClass={cs.LO},
      url={https://arxiv.org/abs/2511.04577}, 
}

@InProceedings{DalGreOli21,
author="Dalmonte, Tiziano
and Grellois, Charles
and Olivetti, Nicola",
editor="Das, Anupam
and Negri, Sara",
title= {{Terminating calculi and countermodels for constructive modal logics}},
booktitle="Automated Reasoning with Analytic Tableaux and Related Methods",
year="2021",
publisher="Springer International Publishing",
address="Cham",
pages="391--408"
}

@misc{DasGroShi25-blog,
  author    = {Das, Anupam and
               Groot, {Jim de} and
               Shillito, Ian},
  title     = {Diamond-free parts of intuitionistic modal logics},
  year      = {2025},
  note      = {Post on The Proof Theory Blog, accessed: 13/02/2026},
  url       = {https://prooftheory.blog/2025/10/03/diamond-free-parts-of-intuitionistic-modal-logics/},
}

@misc{DasMar22-blog,
  author    = {Das, Anupam and
               Marin, Sonia},
  title     = {{B}rouwer meets {K}ripke: constructivising modal logic},
  year      = {2022},
  note      = {Post on The Proof Theory Blog, accessed: 13/02/2026},
  url       = {https://prooftheory.blog/2022/08/19/brouwer-meets-kripke-constructivising-modal-logic/},
}

@inproceedings{DasMar23,
  author       = {Anupam Das and
                  Sonia Marin},
  editor       = {R. Ramanayake and
                  J. Urban},
  title        = {On Intuitionistic Diamonds (and Lack Thereof)},
  booktitle    = {Proceedings~{TABLEAUX} 2023},
  OPTseries       = {Lecture Notes in Computer Science},
  OPTvolume       = {14278},
  pages        = {283--301},
  year         = {2023},
  doi          = {10.1007/978-3-031-43513-3\_16},
}

@article{DerMan,
author = {Dershowitz, Nachum and Manna, Zohar},
title = {Proving Termination with Multiset Orderings},
year = {1979},
volume = {22},
number = {8},
issn = {0001-0782},
doi = {10.1145/359138.359142},
journal = {Communications of the ACM},
pages = {465–476},
}

@article{Dyc92,
    author={Roy Dyckhoff},
    year={1992},
    title={Contraction-free sequent calculi for intuitionistic logic}, 
    journal={Journal of Symbolic Logic},
    volume={57}, 
    number={3},  
    pages={795–807},
doi={doi.org/10.2307/2275431}
}

@article{DycNeg00,
 author = {Roy Dyckhoff and Sara Negri},
 journal = {The Journal of Symbolic Logic},
 number = {4},
 pages = {1499--1518},
 publisher = {[Association for Symbolic Logic, Cambridge University Press]},
 title = {Admissibility of Structural Rules for Contraction-Free Systems of Intuitionistic Logic},
 volume = {65},
 year = {2000},
  url          = {https://doi.org/10.2307/2695061}
}

@article{Ewald1986, 
title={Intuitionistic tense and modal logic}, 
volume={51}, 
DOI={10.2307/2273953}, 
number={1}, 
journal={The Journal of Symbolic Logic}, 
publisher={Cambridge University Press}, 
author={Ewald, W. B.}, 
year={1986}, 
pages={166–179}
}

@inproceedings{Fer23,
author = {F\'{e}r\'{e}e, Hugo and van Gool, Sam},
title = {Formalizing and Computing Propositional Quantifiers},
year = {2023},
isbn = {9798400700262},
publisher = {Association for Computing Machinery},
doi = {10.1145/3573105.3575668},
booktitle = {Proceedings of the 12th ACM SIGPLAN International Conference on Certified Programs and Proofs},
pages = {148–158},
numpages = {11},
location = {Boston, MA, USA},
series = {CPP 2023}
}

@inproceedings{FerGieGooShi24,
title = {{Mechanised uniform interpolation for modal logics K, GL, and iSL}},
author = "Hugo F{\'e}r{\'e}e and {\noopsort{Giessen}}{van der Giessen}, Iris and {\noopsort{Gool}}{van Gool}, Sam and Ian Shillito",
year = "2024",
month = jul,
day = "2",
doi = "10.1007/978-3-031-63501-4\_3",
language = "English",
isbn = "9783031635007",
volume = "2",
series = "Lecture Notes in Computer Science",
publisher = "Springer",
pages = "43--60",
editor = "Christoph Benzm{\"u}ller and Heule, \{Marijn J. H.\} and Schmidt, \{Renate A.\}",
booktitle = "Automated Reasoning",
}

@inproceedings{FerGooIgl25,
  TITLE = {{Formulas rewritten and normalized computationally, and intuitionistically simplified}},
  AUTHOR = {F{\'e}r{\'e}e, Hugo and van Gool, Sam and Iglesias V{\'a}zquez, Yago},
  URL = {https://hal.science/hal-04859429},
  BOOKTITLE = {{36es Journ{\'e}es Francophones des Langages Applicatifs (JFLA 2025)}},
  ADDRESS = {Roiff{\'e}, France},
  YEAR = {2025},
  MONTH = Jan,
  HAL_ID = {hal-04859429},
  HAL_VERSION = {v1},
}

@article{FischerServi77,
	title = {On Modal Logic with an Intuitionistic Base},
	author = {Gis\`{e}le {\noopsort{Fischer}}{Fischer Servi}},
	pages = {141--149},
	doi = {10.1007/bf02121259},
	publisher = {Springer},
	year = {1977},
	volume = {36},
	journal = {Studia Logica}
}

@article{FischerServi84,
	title = {Axiomatizations for some intuitionistic modal logics},
	author = {Gis\`{e}le {\noopsort{Fischer}}{Fischer Servi}},
	pages = {179--194},
	year = {1984},
	volume = {42},
    issue ={3},
	journal = {Rendiconti del Seminario Matematico dell’ Universit`a Politecnica di
Torino}
}

@phdthesis{Gie22,
    author="{\noopsort{Giessen}}{Van der Giessen}, Iris",
    title={Uniform Interpolation and Admissible Rules. Proof-theoretic investigations into (intuitionistic) modal logics},
    series={Quaestiones Infinitae},
    volume={138},
school={Utrecht University},
    year={2022},
url={https://dspace.library.uu.nl/bitstream/handle/1874/423244/proefschrift%20-%206343c2623d6ab.pdf},
}

@inproceedings{GorRamShi21,
  author    = {Rajeev Gor{\'{e}} and
               Revantha Ramanayake and
               Ian Shillito},
  editor    = {Anupam Das and
               Sara Negri},
  title     = {Cut-Elimination for Provability Logic by Terminating Proof-Search:
               Formalised and Deconstructed Using {C}oq},
  booktitle = {Automated Reasoning with Analytic Tableaux and Related Methods - 30th
               International Conference, {TABLEAUX} 2021},
  series    = {Lecture Notes in Computer Science},
  volume    = {12842},
  pages     = {299--313},
  publisher = {Springer},
  year      = {2021},
  doi       = {10.1007/978-3-030-86059-2\_18}
}

@inproceedings{GorShi22,
  author       = {Rajeev Gor{\'{e}} and Ian Shillito},
  editor       = {David Fern{\'{a}}ndez{-}Duque and
                  Alessandra Palmigiano and
                  Sophie Pinchinat},
  title        = {{Direct elimination of additive-cuts in GL4ip: verified and extracted}},
  booktitle    = {Advances in Modal Logic, AiML 2022, Rennes, France, August 22-25,
                  2022},
  pages        = {429--449},
  publisher    = {College Publications},
  year         = {2022},
  url          = {http://www.aiml.net/volumes/volume14/26-Gore-Shillito.pdf},
  bibsource    = {dblp computer science bibliography, https://dblp.org}
}

@inproceedings{GroShiClo25,
  author={{\noopsort{Groot}}{De Groot}, Jim and Shillito, Ian and Clouston, Ranald},
  booktitle={2025 40th Annual ACM/IEEE Symposium on Logic in Computer Science (LICS)}, 
  title={{Semantical analysis of intuitionistic modal logics between CK and IK}}, 
  year={2025},
  volume={},
  number={},
  pages={169-182},
  doi={10.1109/LICS65433.2025.00020}
}

@article{Hud93,
    author = {Hudelmaier, J{\¨{o}}rg},
    title = "{An O(n log n)-space decision procedure for intuitionistic propositional logic}",
    journal = {Journal of Logic and Computation},
    volume = {3},
    number = {1},
    pages = {63-75},
    year = {1993},
    month = {02},
    issn = {0955-792X},
    doi = {10.1093/logcom/3.1.63},
}

@article{Iem18,
    author = {Iemhoff, Rosalie},
    title = {Terminating sequent calculi for two intuitionistic modal logics},
    journal = {Journal of Logic and Computation},
    volume = {28},
    number = {7},
    pages = {1701-1712},
    year = {2018},
    month = {10},
    issn = {0955-792X},
    doi = {10.1093/logcom/exy026}
}

@article{Iem19a,
	Author = {R. Iemhoff},
	Journal = {Archive for Mathematical Logic},
	Number = {1-2},
	Pages = {155-181},
	Title = {Uniform interpolation and sequent calculi in modal logic},
	Volume = {58},
	Year = {2019}}

@article{Iemhoff19b,
title = {Uniform interpolation and the existence of sequent calculi},
journal = {Annals of Pure and Applied Logic},
volume = {170},
number = {11},
pages = {102711},
year = {2019},
issn = {0168-0072},
doi = {10.1016/j.apal.2019.05.008},
author = {Rosalie Iemhoff},
}

@article{Kav16,
  author       = {G. A. Kavvos},
  title        = {The Many Worlds of Modal {\(\lambda\)}-calculi: {I}. {C}urry-{H}oward for
                  Necessity, Possibility and Time},
  journal      = {CoRR},
  volume       = {abs/1605.08106},
  year         = {2016},
  url          = {http://arxiv.org/abs/1605.08106},
  eprinttype    = {arXiv},
  eprint       = {1605.08106}
}

@misc{Kur26,
  author = {Taishi Kurahashi},
  title = {Interpolation properties for several logics},
  howpublished = {\url{https://www2.kobe-u.ac.jp/~tk/jp/notes/ULIP.html}},
  note = {accessed: 13/02/2026}
}

@article{Mendler11,
title = {{Cut-free Gentzen calculus for multimodal CK}},
journal = {Information and Computation},
volume = {209},
number = {12},
pages = {1465-1490},
year = {2011},
note = {Intuitionistic Modal Logic and Applications (IMLA 2008)},
issn = {0890-5401},
doi = {https://doi.org/10.1016/j.ic.2011.10.003},
author = {Michael Mendler and Stephan Scheele}
}

@inproceedings{MenPai05,
  title={Constructive {CK} for contexts},
  author={Mendler, M. and
          de Paiva, V.},
  booktitle={Proceedings {CRR} 2005},
  year={2005}
}

@article{Pit92,
  author    = {Andrew M. Pitts},
  title     = {On an Interpretation of Second Order Quantification in First Order
               Intuitionistic Propositional Logic},
  journal   = {Journal of Symbolic Logic},
  volume    = {57},
  number    = {1},
  pages     = {33--52},
  year      = {1992},
  doi       = {10.2307/2275175}
}

@inproceedings{Plotkin_Stirling86,
author = {Plotkin, Gordon and Stirling, Colin},
title = {A framework for intuitionistic modal logics: extended abstract},
year = {1986},
isbn = {0934613049},
publisher = {Morgan Kaufmann Publishers Inc.},
address = {San Francisco, CA, USA},
booktitle = {Proceedings of the 1986 Conference on Theoretical Aspects of Reasoning about Knowledge},
pages = {399–406},
numpages = {8},
location = {Monterey, California},
series = {TARK '86}
}

@book{Prior57TimeModality,
  title={Time and Modality},
  author={Prior, A. N.},
  lccn={57002267},
  series={John Locke Lectures},
  year={1957},
  publisher={Clarendon Press}
}

@misc{Rocq,
  author       = {{The Rocq Development Team}},
  title        = {{The Rocq Prover}},
  month        = apr,
  year         = 2025,
  publisher    = {Zenodo},
  version      = {9.0},
  doi          = {10.5281/zenodo.15149629}
}

@phdthesis{Shi23,
  author = {Shillito, I.},
  title = {New Foundations for the Proof Theory of Bi-Intuitionistic and Provability Logics Mechanized in {C}oq},
  school  = {Australian National University},
  address = {Canberra},
  year = {2023},
  url = {http://www.proquest.com/docview/2812065824?pq-origsite=gscholar&fromopenview=true}
}

@InProceedings{ShiGieGorIem23,
author="Shillito, Ian
and Iris {\noopsort{Giessen}}{van der Giessen}
and Gor{\'e}, Rajeev
and Iemhoff, Rosalie",
editor="Ramanayake, Revantha
and Urban, Josef",
title="A New Calculus for Intuitionistic Strong {L}{\"o}b Logic: Strong Termination and Cut-Elimination, Formalised",
booktitle="Automated Reasoning with Analytic Tableaux and Related Methods, TABLEAUX 2023",
year="2023",
publisher="Springer Nature Switzerland",
address="Cham",
pages="73--93",
isbn="978-3-031-43513-3",
doi={10.1007/978-3-031-43513-3_5}
}

@phdthesis{Simpson94PhD,
  author    = {Alex K. Simpson},
  title     = {The proof theory and semantics of intuitionistic modal logic},
  school    = {University of Edinburgh},
  year      = {1994}
}

@incollection{Vis1996,
  Author = {Visser, A.},
  Booktitle = {{G{\"o}del '96 proceedings}},
  Editor = {H\'ajek, P.},
  Pages = {139-164},
  Publisher = {{Springer-Verlag}},
  Series = {Lecture Notes in Logic},
  Title = {{Uniform interpolation and layered bisimulation}},
  Url = {http://projecteuclid.org/download/pdf_1/euclid.lnl/1235417019},
  Volume = {6},
  Year = {1996}}

@article{Vor58,
 author = {Vorob'ev, N. N.},
 title = {{A new algorithm of derivability in a constructive calculus of statements}},
 journal = {Trudy Matematicheskogo Instituta imeni V. A. Steklova},
 publisher = {Acad. Sci. USSR},
 issn = {0371-9685},
 volume = {52},
 pages = {193--225},
 year = {1958},
 language = {Russian},
 note = {(Translation in American Mathematical Society Translations, series 2, volume 94, 1970)}
}

@article{Wijesekera90,
title = "Constructive modal logics {I}",
journal = "Annals of Pure and Applied Logic",
volume = "50",
number = "3",
pages = "271--301",
year = "1990",
doi = "10.1016/0168-0072(90)90059-B",
author = "D. Wijesekera",
}

@article{WijesekeraNerode05,
title = {Tableaux for constructive concurrent dynamic logic},
journal = {Annals of Pure and Applied Logic},
volume = {135},
number = {1},
pages = {1-72},
year = {2005},
issn = {0168-0072},
doi = {https://doi.org/10.1016/j.apal.2004.12.001},
author = {Duminda Wijesekera and Anil Nerode}
}

@Incollection{WolterZakharyaschev99b,
author="Wolter, Frank
and Zakharyaschev, Michael",
editor="Cantini, Andrea
and Casari, Ettore
and Minari, Pierluigi",
title="Intuitionistic Modal Logic",
bookTitle="Logic and Foundations of Mathematics: {S}elected Contributed Papers of the Tenth International Congress of Logic, Methodology and Philosophy of Science, Florence, August 1995",
year="1999",
publisher="Springer Netherlands",
address="Dordrecht",
pages="227--238",
isbn="978-94-017-2109-7",
doi="10.1007/978-94-017-2109-7\_17"
}

@article{WolterZakharyaschev97,
author = {Wolter, Frank and Zakharyaschev, Michael},
year = {1997},
pages = {73--92},
title = {On the Relation Between Intuitionistic and Classical Modal Logics},
volume = {36},
journal = {Algebra and Logic},
doi = {10.1007/BF02672476}
}

@article{WolterZakharyaschev99a,
author = {Wolter, Frank and Zakharyaschev, Michael},
year = {1999},
pages = {168--186},
title = {Intuitionistic Modal Logics as Fragments of Classical Bimodal Logics},
journal = {Logic at Work, Studies in Fuzziness Soft Computing},
volume = {24}
}

\end{document}